\documentclass{irmaart}
\input{ha.sty}
\usepackage{graphicx}
\usepackage{gastex}
\usepackage[displaymath]{lineno}
\usepackage[hypertex,hyperindex,pagebackref,final]{hyperref}
\usepackage{calc}
\newcounter{hours}\newcounter{minutes}
\newcommand\printtime{\setcounter{hours}{\time/60}%
  \setcounter{minutes}{\time-\value{hours}*60}%
  \thehours\,h\,\theminutes}
\newcommand\dateandtime{\today\quad\printtime}
\usepackage[boxed]{algorithm}    
\usepackage[noend]{algorithmic}  
  \newlength\commentspace
  \newcommand\algcomment[2]{\makebox[0pt][l]{\hspace{-#1em}%
    \hspace{\commentspace}$\triangleright$ #2}}
\newcommand{\algorithmicfunc}[1]{\textsc{#1}}
\newcommand{\FUNC}[1]{\item[\algorithmicfunc{#1}]}
\usepackage{pst-all}
 \newpsobject{showgrid}{psgrid}{%
   subgriddiv=1,griddots=10,gridlabels=6pt}
\newcommand{\eqed}{\tag*{\qedsymbol}}
\newcommand{\Card}{\operatorname{Card}}
\newcommand\A{\mathcal{A}}
\newcommand\B{\mathcal{B}}

\newcommand{\cF}{\mathcal{F}}
\newcommand{\cK}{\mathcal{K}}
\newcommand{\cM}{\mathcal{M}}
\newcommand{\cP}{\mathcal{P}}
\newcommand{\cQ}{\mathcal{Q}}

\newcommand{\cU}{\mathcal{U}}
\newcommand{\cW}{\mathcal{W}}
\newcommand\e{\varepsilon}
\let\tto\xrightarrow
\begin{document}

\markboth{J.~Berstel, L.~Boasson, O.~Carton, I.~Fagnot}{Minimization of automata}
\title{Minimization of automata} 
\author{Jean Berstel$^1$,
  Luc Boasson$^2$,
  Olivier Carton$^2$, Isabelle~Fagnot$^{1,*}$}
\address{$^1$Laboratoire d'Informatique Gaspard-Monge\\
  Universit\'e Paris-Est Marne-la-Vall\'ee\\
  5, boulevard Descartes, Champs-sur-Marne,
  F-77454 Marne-la-Vall\'ee Cedex 2\\[2mm]
  $^2$LIAFA\\
  $^*$Universit\'e Paris Diderot\\
  Case 7014,
  F-75205 Paris Cedex 13\\[2mm]
  email:\,\url{{berstel,fagnot}@univ-mlv.fr, 
    {boasson,carton}@liafa.jussieu.fr}\\[4mm]
  \upshape{\dateandtime}}

\maketitle\label{chapterBBC} 



\localtableofcontents

%
%

\section{Introduction}
This chapter is concerned with the design and analysis of algorithms
for minimizing finite automata. Getting a minimal automaton is a
fundamental issue in the use and implementation of finite automata
tools in frameworks like text processing, image analysis, linguistic
computer science, and many other applications.

There are two main families of minimization algorithms. The first by a
sequence of refinements of a partition of the set of states, the
second by a sequence of fusions or merges of states.  Among the
algorithms of the first family, we mention a simple algorithm described in
the book~\cite{Hopcroft&Ullman:1979}. It operates by a traversal of
the product of the automaton with itself, and therefore is in time and
space complexity $O(n^2)$. Other algorithms are Hopcroft's and Moore's
algorithms, which will be considered in depth later.  The linear-time
minimization of acyclic automata of Revuz belongs to the second
family. Brzozowski's algorithm stands quite isolated and fits in
neither of these two classes.

The algorithm for the minimization of complete deterministic finite
state automata given by Hopcroft~\cite{Hopcroft:1971} runs in
worst-case time $O(n\log n)$. It is, up to now, the most efficient
algorithm known in the general case. It has recently been extended to
incomplete deterministic finite automata
\cite{Valmari&Lehtinen:2008},\cite{Beal&Crochemore:2008}.

Hopcroft's algorithm is related to Moore's partition refinement
algorithm \cite{Moore:1956}, although it is different.  One of the
purposes of this text is the comparison of the nature of Moore's
and Hopcroft's algorithms. This gives some new insight into both
algorithms. As we shall see, these algorithms are quite different both
in behavior and in complexity. In particular, we show that it is not
possible to simulate the computations of one algorithm by the
other.

Moore's partition refinement algorithm is much simpler than Hopcroft's
algorithm. It has been shown \cite{Bassino&David&Nicaud:2009} that,
although its worst-case behavior is quadratic, its average running
time is $O(n\log n)$. No evaluation of the average is known for
Hopcroft's algorithm.

The family of algorithms based on fusion of states is important in
practice for the construction of minimal automata representing finite
sets, such as dictionaries in natural language processing. A linear
time implementation of such an algorithm for cycle-free automata was
given by Revuz~\cite{Revuz:1992}. This algorithm has been extended to
a more general class of automata by Almeida and
Zeitoun~\cite{Almeida&Zeitoun:2008}, namely to automata where all
strongly connected components are simple cycles. It has been
demonstrated in~\cite{Beal&Crochemore:2007} that minimization by state
fusion, which is not always possible, works well for local automata.

There is another efficient incremental algorithm for finite sets, by
Daciuk \textit{et al.}  \cite{Daciuk&Mihov&Watson&Watson:2000}. The
advantage of this algorithm is that it does not build the intermediate
trie which is rather space consuming.

We also consider updating a minimal automaton when a word is added or
removed from the set it recognizes.

Finally, we discuss briefly the case of nondeterministic automata. It
is well-known that minimal nondeterministic automata are not
unique. However, there are several subclasses where the minimal
automaton is unique.

We do not consider here the problem of constructing a minimal
automaton starting from another description of the regular language,
such as the synthesis of an automaton from a regular expression. We
also do not consider devices that may be more space efficient, such as
alternating automata or two-way automata.  Other cases not considered
here concern sets of infinite words and the minimization of their
accepting devices.

The chapter is organized as follows. The first section just fixes
notation, the next describes briefly Brzozowski's algorithm. In
Section~\ref{BBC:sec:moore}, we give basic facts on Moore's
minimization algorithm. Section~\ref{BBC:sec:hopcroft} is a detailed
description of Hopcroft's algorithm, with the proof of correctness and
running time. It also contains the comparison of Moore's and
Hopcroft's algorithms. The next section is devoted to so-called slow
automata. Some material in these two sections is new.

Sections~\ref{BBC:sec:fusion} and~\ref{BBC:sec:dynamic} are devoted to the family of algorithms
working by fusion. We describe in particular Revuz's algorithm and its
generalization by Almeida and Zeitoun, the incremental algorithm of
Daciuk \emph{et al.}, and dynamic minimization. The last section
contains miscellaneous results on special cases and a short discussion
of nondeterministic minimal automata.
%
%

\section{Definitions and notation}

It appears to be useful, for a synthetic presentation of the
minimization algorithms of Moore and of Hopcroft, to introduce some
notation for partitions of the set of states. This section just fixes
this notation. 

\paragraph*{Partitions and equivalence relations.}

A \emph{partition}\index{partition} of a set $E$ is a family $\cP$ of
nonempty, pairwise disjoint subsets of $E$ such that
$E=\bigcup\limits_{P\in\cP}P$. The \emph{index}\index{index of a
  partition} of the partition is the number of its elements.  A
partition defines an equivalence relation $\equiv_\cP$ on $E$.
Conversely, the set of all equivalence classes $[x]$, for $x\in E$, of
an equivalence relation on $E$ defines a partition of $E$.  This is
the reason why all terms defined for partitions have the same meaning
for equivalence relations and vice versa.

A subset $F$ of $E$ is \emph{saturated}\index{saturated subset} by
$\cP$ if it is the union of classes of $\cP$.  Let $\cQ$ be another
partition of $E$. Then $\cQ$ is a
\emph{refinement}\index{partition!refinement} of $\cP$, or $\cP$ is
\emph{coarser}\index{coarser partition}\index{partition!coarser} than
$\cQ$, if each class of $\cQ$ is contained in some class of $\cP$.  If
this holds, we write $\cQ\le\cP$. The index of $\cQ$ is then larger
than the index of $\cP$.

Given two partitions $\cP$ and $\cQ$ of a set $E$, we denote by
$\cU=\cP\wedge\cQ$ the coarsest partition which refines $\cP$ and
$\cQ$. The classes of $\cU$ are the nonempty sets $P\cap Q$, for
$P\in\cP$ and $Q\in\cQ$. The notation is extended to a set of
partitions in the usual way: we write
$\cP=\cP_1\wedge\cdots\wedge\cP_n$ for the common refinement of
$\cP_1,\ldots,\cP_n$. If $n=0$, then $\cP$ is the universal partition
of $E$ composed of the single class $E$. This partition is the neutral
element for the $\wedge$-operation.

Let $F$ be a subset of $E$. A partition $\cP$ of $E$ induces a
partition $\cP'$ of $F$ by intersection: $\cP'$ is composed of the
nonempty sets $P\cap F$, for $P\in\cP$. If $\cP$ and $\cQ$ are
partitions of $E$ and $\cQ\le\cP$, then the restrictions $\cP'$ and
$\cQ'$ to $F$ still satisfy $\cQ'\le\cP'$.

If $\cP$ and $\cP'$ are partitions of disjoint sets $E$ and $E'$, we
denote by $\cP\vee\cP'$ the partition of $E\cup E'$ whose restriction
to $E$ and $E'$ are $\cP$ and $\cP'$ respectively. So, one may write
\begin{displaymath}
  \cP=\bigvee_{P\in\cP}\{P\}\,.
\end{displaymath}

\paragraph*{Minimal automaton.}
We consider a deterministic automaton $\A=(Q,i,F)$ over the alphabet $A$ with
set of states $Q$, initial state $i$, and set of final states $F$.  To
each state $q$ corresponds a subautomaton of $\A$ obtained when $q$ is
chosen as the initial state. We call it the \textit{subautomaton
  rooted at}\index{subautomaton rooted at a state} $q$ or simply the
automaton at $q$. Usually, we consider only the trim part of this
automaton.  To each state $q$ corresponds a language $L_q(\A)$ which
is the set of words recognized by the subautomaton rooted at $q$, that
is
\begin{displaymath}
  L_q(\A)=\{w\in A^*\mid q\cdot w\in F\}\,.
\end{displaymath}
This language is called the \emph{future}\index{future of a
  state}\index{state!future} of the state $q$, or also the \emph{right
  language}\index{right language}\index{language!right} of this state.
Similarly one defines the \emph{past}\index{past of a
  state}\index{state!past} of $q$, also called the \emph{left
  language}\index{left language}\index{language!left}, as the set
$\{w\in A^*\mid i\cdot w= q\}$.  The automaton $\A$ is
\emph{minimal}\index{minimal automaton}\index{automaton!minimal}
if $L_p(\A)\ne L_q(\A)$ for each pair of distinct states $p,q$. The equivalence
relation $\equiv$ defined by
\begin{displaymath}
  p\equiv q \quad\text{if and only if}\quad L_p(\A)= L_q(\A)
\end{displaymath}
is a \emph{congruence}, that is $ p\equiv q$ implies $p\cdot a\equiv
q\cdot a$ for all letters $a$. It is called the \emph{Nerode
  congruence}\index{Nerode congruence}\index{equivalence!Nerode}. Note
that the Nerode congruence saturates the set of final states.  Thus an
automaton is minimal if and only if its Nerode equivalence is the
identity.

Minimizing an automaton is the problem of computing the Nerode
equivalence. Indeed, the \emph{quotient}\index{quotient
  automaton}\index{automaton!quotient} automaton $\A/{\equiv}$
obtained by taking for set of states the set of equivalence classes of
the Nerode equivalence, for the initial state the class of the initial
state $i$, for set of final states the set of equivalence
classes of states in $F$ and by defining the transition function by
$[p]\cdot a=[p\cdot a]$ accepts the same language, and its Nerode
equivalence is the identity. The minimal automaton recognizing a given
language is unique.

\paragraph*{Partitions and automata.}

Again, we fix a deterministic automaton $\A=(Q,i,F)$ over the alphabet
$A$.  It is convenient to use the shorthand $P^c$ for $Q\setminus P$
when $P$ is a subset of the set~$Q$.

Given a set $P\subset Q$ of states and a letter $a$, we denote by
$a^{-1}P$ the set of states $q$ such that $q\cdot a\in P$. Given sets
$P,R\subset Q$ and $a\in A$, we denote by
\begin{displaymath}
  (P,a)|R
\end{displaymath}
the partition of $R$ composed of the nonempty sets among the two sets
\begin{displaymath}
  R\cap a^{-1}P=\{q\in R\mid q\cdot a\in P\}
\quad\text{and}\quad 
R\setminus a^{-1}P=\{q\in R\mid q\cdot a\notin P\}\,.
\end{displaymath}
Note that $R\setminus a^{-1}P=R\cap (a^{-1}P)^c=R\cap a^{-1}(P^c)$ so
the definition is symmetric in $P$ and $P^c$. In particular
\begin{equation}
(P,a)|R=(P^c,a)|R\,.\label{BBC:eq:0}
\end{equation}
The pair $(P,a)$ is called a \emph{splitter}\index{splitter}\index{automaton!splitter}.
Observe that $(P,a)|R=\{R\}$ if either $R\cdot a\subset P$ or $R\cdot
a\cap P=\emptyset$, and $(P,a)|R$ is composed of two classes if both
$R\cdot a\cap P\ne\emptyset$ and $R\cdot a\cap P^c\ne\emptyset$ or
equivalently if $R\cdot a\not\subset P^c$ and $R\cdot a\not\subset P$.
If $(P,a)|R$ contains two classes, then we say that $(P,a)$
\emph{splits} $R$. Note that the pair $S=(P,a)$ is called a
splitter even if it does not split.

It is useful to extend the notation above to words. Given a word $w$
and sets $P,R\subset Q$ of states, we denote by $w^{-1}P$ the set of
states such that $q\cdot w\in P$, and by $(P,w)|R$ the partition of
$R$ composed of the nonempty sets among
\begin{displaymath}
  R\cap w^{-1}P=\{q\in R\mid q\cdot w\in P\}
  \quad\text{and}\quad
  R\setminus w^{-1}P=\{q\in R\mid q\cdot w\notin P\}\,.
\end{displaymath}
As an example, the partition $(F,w)|Q$ is the partition of $Q$ into
the set of those states from which $w$ is accepted, and the other
ones. A state $q$ in one of the sets and a state $q'$ in the other are
sometimes called \emph{separated}\index{separated
  states}\index{state!separated} by $w$.

The Nerode equivalence is the coarsest equivalence relation on the set
of states that is a (right) congruence saturating $F$. With the notation of
splitters, this can be rephrased as follows.

\begin{proposition}\label{BBC:prop:Nerode}
  The partition corresponding to the Nerode equivalence is the
  coarsest partition $\cP$ such that no splitter $(P,a)$, with
  $P\in\cP$ and $a\in A$, splits a class in $\cP$, that is such that
  $(P,a)|R=\{R\}$ for all $P,R\in\cP$ and $a\in A$.\qed
\end{proposition}

We use later the following lemma which is already given in
Hopcroft's paper \cite{Hopcroft:1971}. It is the basic observation
that  ensures that Hopcroft's algorithm works correctly.
\begin{lemma}\label{BBC:le:magique}
  Let $P$ be a set of states, and let $\cP=\{P_1,P_2\}$ be a partition of
  $P$. For any letter $a$ and for any set of states  $R$, one has
  \begin{displaymath}
    (P,a)|R\wedge(P_1,a)|R=(P,a)|R\wedge(P_2,a)|R=(P_1,a)|R\wedge(P_2,a)|R\,,
  \end{displaymath}
  and consequently
  \begin{gather}
    (P,a)|R\ge (P_1,a)|R\wedge (P_2,a)|R\label{BBC:eq:magique1}\,,\\
    (P_1,a)|R\ge (P,a)|R\wedge (P_2,a)|R\label{BBC:eq:magique2}\,.
  \end{gather}
\end{lemma}


%
%

\section{Brzozowski's algorithm} \label{BBC:sec:brzozo}

The minimization algorithm given by Brzozowski
\cite{Brzozowski:1963} is quite different from the two families of
iterative algorithms (by refinement and by fusion) that we consider in
this chapter. Although its worst-case behavior is exponential, it is
conceptually simple, easy to implement, and it is quite efficient in
many cases. Moreover, it does not require the  automaton to be
deterministic, contrary to the algorithms described later.

Given an automaton $\A=(Q,I,F,E)$ over an alphabet $A$, its
\emph{reversal}\index{automaton!reversal}\index{reversal!automaton} is
the automaton denoted $A^R$ obtained by exchanging the initial and the
final states, and by inverting the orientation of the edges.
Formally $\A^R=(Q,F,I,E^R)$, where $E^R=\{(p,a,q)\mid (q,a,p)\in E\}$.
The basis for Brzozowski's algorithm is the following result.

\begin{figure}
\centering\scalebox{0.85}{%
  \begin{picture}(60,20)(0,-5)
    \gasset{Nh=6,Nadjust=w,Nadjustdist=2}
    \node[Nmarks=i,iangle=-90](0)(0,0){$0$}
    \node(1)(15,0){$1$}
    \node(2)(30,0){$2$}
    \node[Nframe=n](cdot)(45,0){$\cdots$}
    \node[Nmarks=f,fangle=-90](n)(60,0){$n$}
    \drawedge(0,1){$a$}
    \drawedge(1,2){$a,b$}
    \drawedge(2,cdot){$a,b$}
    \drawedge(cdot,n){$a,b$}
    \drawloop[loopangle=90](0){$a,b$}
       \end{picture}
\qquad\qquad
  \begin{picture}(60,20)(0,-5)
    \gasset{Nh=6,Nadjust=w,Nadjustdist=2}
    \node[Nmarks=f,fangle=-90](0)(0,0){$0$}
    \node(1)(15,0){$1$}
    \node(2)(30,0){$2$}
    \node[Nframe=n](cdot)(45,0){$\cdots$}
    \node[Nmarks=i,iangle=-90](n)(60,0){$n$}
    \drawedge(1,0){$a$\vphantom{,b}}
    \drawedge(2,1){$a,b$}
    \drawedge(cdot,2){$a,b$}
    \drawedge(n,cdot){$a,b$}
    \drawloop[loopangle=90](0){$a,b$}
  \end{picture}}

  \caption{The automaton on the left recognizing the language
    $A^*aA^n$. It has $n+1$ states and the minimal deterministic
  automaton for this language has $2^n$ states. The automaton on the
  right is its reversal. It is minimal and recognizes $A^naA^*$.}
  \label{BBC:fig:nondetA}
\end{figure}
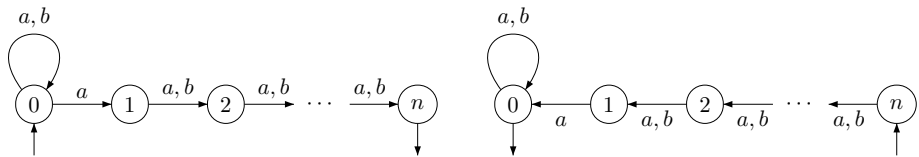

\begin{proposition}
  Let $\A$ be a finite deterministic automaton, and let $\A^\sim$ be the
  deterministic trim automaton obtained by determinizing and trimming
  the reversal $\A^R$. Then $\A^\sim$ is minimal.
\end{proposition}

For a proof of this proposition, see for instance Sakarovitch's book
\cite{Sakarovitch:2009}.  The minimization algorithm now is just a
double application of the operation. Observe that the automaton one
starts with need not to be deterministic.

\begin{corollary}
  Let $\A$ be a finite automaton. Then $(\A^\sim)^\sim$ is the minimal
  automaton of $\A$.
\end{corollary}

\begin{example}
  We consider the automata given in Figure~\ref{BBC:fig:nondetA} over
  the alphabet $A=\{a,b\}$.  Each automaton is the reversal of the
  other. However, determinization of the automaton on the left
  requires exponential time and space.

\end{example}

%
%

\section{Moore's algorithm} \label{BBC:sec:moore}

The minimization algorithm given by Moore \cite{Moore:1956} computes
the Nerode equivalence by a stepwise refinement of some initial
equivalence. All automata are assumed to be deterministic.

\subsection{Description}

Let $\A=(Q,i,F)$ be an automaton over an alphabet $A$. Define, for
$q\in Q$ and $h\ge0$, the set
\begin{displaymath}
  L_q^{(h)}(\A)=\{w\in A^*|\ |w|\le h,\ q\cdot w\in F\}\,.
\end{displaymath}
The \emph{Moore equivalence}\index{equivalence!Moore}\index{Moore
  equivalence} of order $h$ is the equivalence $\equiv_h$ on $Q$
defined by
\begin{displaymath}
p\equiv_h q\iff L_p^{(h)}(\A)=L_q^{(h)}(\A)\,.
\end{displaymath}
Using the notation of partitions introduced above, one can rephrase the
definitions of the Nerode equivalence and of the Moore equivalence of
order $h$. These are the equivalences defined by
\begin{displaymath}
  \bigwedge_{w\in A^*}(F,w)|Q\,,\quad\text{and}\quad
  \bigwedge_{w\in A^*,\ |w|\le h}(F,w)|Q\,.
\end{displaymath}
Since the set of states is finite, there is a smallest $h$ such that
the Moore equivalence $\equiv_h$ equals the Nerode equivalence
$\equiv$.  We call this integer the \emph{depth}%
\index{depth of Moore's algorithm}\index{Moore algorithm!depth}%
\index{automaton!depth} of Moore's algorithm on the finite automaton
$\A$, or the depth of $\A$ for short.  The depth depends in fact only
on the language recognized by the automaton, and not on the particular
automaton under consideration. Indeed, each state of an automaton
recognizing a language $L$ represents in fact a left quotient
$u^{-1}L$ for some word $u$.

The depth is the smallest $h$ such that $\equiv_h$ equals
$\equiv_{h+1}$.  This leads to the refinement algorithm that computes
successively $\equiv_0$, $\equiv_1$, \dots, $\equiv_h$, \dots, halting
as soon as two consecutive equivalences are equal. The next property
gives a method to compute the Moore equivalences efficiently.

\begin{proposition}\label{BBC:prop:Mooreiteratif}
  For two states $p, q\in Q$, and $h\ge 0$, one has
  \begin{equation}
    \label{BBC:eq:Mooreiteratif}
    p\equiv_{h+1}q\quad\iff\quad p\equiv_h q\ \ \text{and}\ \ p\cdot a\equiv_h
    q\cdot a \text{ for all $a\in A$}\,.
  \end{equation}
\end{proposition}


We use this proposition in a slightly different formulation. Denote by
$\cM_h$ the partition corresponding to the Moore equivalence
$\equiv_h$. Then the following equations hold.
\begin{proposition}\label{BBC:prop:Moore}
  For $h\ge0$, one has
  \begin{displaymath}
    \cM_{h+1}= \cM_h\wedge\bigwedge_{a\in
      A}\bigwedge_{P\in\cM_h}(P,a)|Q\,
    = \bigvee_{R\in\cM_{h}}
    \Bigl(\bigwedge_{a\in A}\bigwedge_{P\in\cM_{h}}(P,a)|R\Bigr)\,.\eqed
  \end{displaymath}
\end{proposition}

\begin{figure}
  \hrule\smallskip\par
  \setlength{\commentspace}{5cm}
  \begin{algorithmic}
    \FUNC{Moore$(\A)$}
      \STATE\algcomment{-1}{The initial partition}$\cP\gets\{F, F^c\}$ 
      \REPEAT
       \STATE\algcomment{0}{$\cP'$ is the old partition,
         $\cP$ is  the new one}$\cP'\gets\cP$
       \FORALL{$a \in A$} 
         \STATE $\cP_a\gets \bigwedge_{P\in\cP}(P,a)|Q$
       \ENDFOR
       \STATE $\cP\gets\cP\wedge\bigwedge_{a\in A}\cP_a$
     \UNTIL $\cP=\cP'$
\end{algorithmic}
\hrule\smallskip\par
\caption{Moore's minimization algorithm.}\label{BBC:alg:Moore}
\end{figure}

\noindent 
The computation is described in Figure~\ref{BBC:alg:Moore}. It is
realized by a loop that refines the current partition. The computation
of the refinement of $k$ partitions of a set swith $n$ elements can be
done in time $O(kn^2)$ by brute force. A radix sort improves the
running time to $O(kn)$. With $k=\Card(A)$, each tour in the loop is
realized in time $O(kn)$, so the total time is $O(\ell kn)$, where
$\ell$ is the number of refinement steps in the computation of the
Nerode equivalence $\equiv$, that is the depth of the automaton.

The worst case behavior is obtained for $\ell=n-2$. We say that
automata having maximal depth are \emph{slow}\index{slow
  automaton}\index{automaton!slow} and more precisely are \emph{slow
  for Moore}\index{slow for Hopcroft}\index{Hopcroft!automaton slow}
automata. These automata are investigated later. We will show that
they are equivalent to automata we call \textit{slow for
  Hopcroft}\index{slow for Moore}\index{Moore!automaton slow}.

\subsection{Average complexity}

The average case behavior of Moore's algorithm has recently been
studied in several papers. We report here some results given in
\cite{Bassino&David&Nicaud:2009,David:2010}. The authors make a
detailed analysis of the distribution of the number $\ell$ of
refinement steps in Moore's algorithm, that is of the depth of
automata, and they prove that there are only a few automata for which
this depth is larger than $\log n$.

More precisely, fix some alphabet and we consider deterministic
automata over this alphabet.

A \emph{semi-automaton}\index{semi-automaton}\index{automaton!semi}
$\cK$ is an automaton whose set of final states is not specified.
Thus, an automaton is a pair $(\cK,F)$, where $\cK$ is a
semi-automaton and $F$ is the set of final states. The following
theorem given an upper bound on the average complexity of Moore's
algorithm
for all automata derived from a given semiautomaton. 

\begin{theorem}[Bassino, David, Nicaud
   \cite{Bassino&David&Nicaud:2009}]\label{BBC:BassinoDavidNicaud}
   Let $\cK$ be a semi-automaton with $n$ states. The average
   complexity of Moore's algorithm for the automata $(\cK, F)$, for
   the uniform probability distribution over the sets $F$ of final
   states, is $O(n\log n)$.
\end{theorem}
The result also holds for Bernoulli distributions for final
states. The result remains valid for subfamilies of automata such as
strongly connected automata or group automata.

When all semi-automata are considered to be equally like, then the
following bound is valid.

\begin{theorem}[David \cite{David:2010}]\label{BBC:theoremDavid}
  The average complexity  of Moore's algorithm, for the uniform
  probability over all complete automata with $n$ states, is
  $O(n\log\log n)$.
\end{theorem}

This result is remarkable in view of the lower bound which is given
in the following statement~\cite{Bassino&David&Nicaud:2011}.

\begin{theorem}
  If the underlying alphabet has at least two letters, then Moore's
  algorithm, applied on a minimal automaton with $n$ states, requires
  at least $\Omega(n\log\log n)$ operations.
\end{theorem}

%
%

%
%

\section{Hopcroft's algorithm} \label{BBC:sec:hopcroft}

Hopcroft~\cite{Hopcroft:1971} has given an algorithm that computes the
minimal automaton of a given deterministic automaton.  The running
time of the algorithm is $O(k\, n\log n)$ where $k$ is the
cardinality of the alphabet and $n$ is the number of states of the
given automaton. The algorithm has been described and re-described
several times
\cite{Gries:1973,Aho&Hopcroft&Ullman:1974,Beauquier&Berstel&Chretienne:1992,Blum:1996,Knuutila:2001}.

\subsection{Outline}
The algorithm is outlined in the function \textsc{Hopcroft} given in
Figure~\ref{BBC:hopcroft}.  We denote by $\min(P,P')$ the set of smaller
size of the two sets $P$ and~$P'$, and any one of them if they have the
same size.

\begin{figure}
  \setlength{\commentspace}{7cm}
\hrule\smallskip\par
  \begin{algorithmic}[1]
    \FUNC{Hopcroft$(\A)$}
    \STATE\algcomment{0}{The initial partition}$\cP\gets \{F, F^c\}$
    \STATE\algcomment{0}{The waiting set}$\cW\gets \emptyset$
    \FORALL{$a \in A$} \STATE\algcomment{1}{Initialization of the waiting
      set}$\textsc{Add}((\min(F,F^c),a),\cW)$
    \ENDFOR
    \WHILE{$\cW\neq\emptyset$} \STATE\algcomment{1}{Take and remove some
      splitter}$(W,a) \gets
    \textsc{TakeSome}(\cW)$ \label{BBC:h:takesome}
    \FOR{each $P\in \cP$ which is split by
      $(W,a)$}
    \STATE\algcomment{2}{Compute the split}$P',P'' \gets
    (W,a)|P$ \label{BBC:h:split}
    \STATE\algcomment{2}{Refine the
      partition}\textsc{Replace} $P$ by $P'$ and $P''$ in $\cP$
    \FORALL{\algcomment{4.82}{Update the waiting set}$b\in A$}
    \IF{$(P,b)\in \cW$} \STATE \textsc{Replace} $(P,b)$ by $(P',b)$
    and $(P'',b)$ in $\cW$ \ELSE \STATE
    $\textsc{Add}((\min(P',P''),b),\cW)$
    \ENDIF
    \ENDFOR
    \ENDFOR
    \ENDWHILE
  \end{algorithmic}
\hrule\smallskip\par
\caption{Hopcroft's minimization algorithm.}\label{BBC:hopcroft}
\end{figure}

Given a deterministic automaton~$\mathcal{A}$, Hopcroft's algorithm
computes the coarsest congruence which saturates the set~$F$ of final
states.  It starts from the partition $\{F, F^c\}$ which obviously
saturates~$F$ and refines it until it gets a congruence.  These refinements
of the partition are always obtained by splitting some class into two
classes.


The algorithm proceeds as follows.  It maintains a current partition
$\mathcal{P} = \{P_1,\ldots,P_n\}$ and a current set~$\mathcal{W}$ of
\emph{splitters}\index{Hopcroft's algorithm!splitter}\index{splitter
  in Hopcroft's algorithm}, that is of
pairs $(W,a)$ that remain to be processed, where $W$ is a class
of~$\mathcal{P}$ and $a$ is a letter.  The set~$\mathcal{W}$ is called
the \emph{waiting} set\index{Hopcroft's algorithm!waiting
  set}\index{waiting set in Hopcroft's algorithm}.  The
algorithm stops when the waiting set~$\mathcal{W}$ becomes empty.
When it stops, the partition~$\mathcal{P}$ is the coarsest congruence
that saturates~$F$.  The starting partition is the partition $\{F,
F^c\}$ and the starting set~$\mathcal{W}$ contains all pairs $(\min(F,
F^c),a)$ for $a \in A$.

The main loop of the algorithm removes one splitter $(W,a)$ from the
waiting set~$\mathcal{W}$ and performs the following actions.  Each
class $P$ of the current partition (including the class~$W$) is
checked as to whether it is split by the pair $(W,a)$.  If $(W,a)$
does not split $P$, then nothing is done.  On the other hand, if
$(W,a)$ splits $P$ into say $P'$ and~$P''$, the class $P$ is replaced
in the partition~$\mathcal{P}$ by $P'$ and~$P''$.  Next, for each
letter $b$, if the pair $(P,b)$ is in~$\mathcal{W}$, it is replaced
in~$\mathcal{W}$ by the two pairs $(P',b)$ and $(P'',b)$, otherwise
only the pair $(\min(P',P''),b)$ is added to~$\mathcal{W}$.

It should be noted that the algorithm is not really deterministic because
it has not been specified which pair $(W,a)$ is taken from~$\mathcal{W}$ to
be processed at each iteration of the main loop. This means that for a
given automaton, there are many executions of the algorithm.  It turns out
that all of them produce the right partition of the states.  However,
different executions may give rise to different sequences of splitting and
also to different running time.  Hopcroft has proved that the running time
of any execution is bounded by $O(|A| n \log n)$.


\subsection{Behavior}

The pair $(\cP,\cW)$ composed of the current partition and the current
waiting set in some execution of Hopcroft's algorithm is called a
\emph{configuration}\index{configuration in Hopcroft's
  algorithm}\index{Hopcroft's algorithm!configuration}. The following
proposition describes the evolution of the current partition in
Hopcroft's algorithm. Formula~\ref{BBC:eq:Hopcroftsplitter} is the key
inequality for the proofs of correctness and termination. We will use
it in the special case where the set $R$ is a class of the current partition.

\begin{proposition}\label{BBC:prop:Hopcroftsplitter}
  Let $(\cP,\cW)$ be a configuration 
  in some execution of Hopcroft's algorithm on an automaton $\A$ on $A$.  For
  any $P\in\cP$, any subset $R$  of a class of  $\cP$, and $a\in A$, one has
  \begin{equation}\label{BBC:eq:Hopcroftsplitter}
    (P,a)|R\ge\bigwedge_{(W,a)\in\cW}(W,a)|R\,,
  \end{equation}
  that is, the partition $(P,a)|R$ is coarser than the partition
  $\bigwedge_{(W,a)\in \cW}(W,a)|R$.
\end{proposition}

\begin{proof}
  The proof is by induction on the steps of an execution. The initial
  configuration $(\cP,\cW)$ is composed of the initial partition is
  $\cP=\{F,F^c\}$ and the initial waiting set $\cW$ is either
  $\cW=\{(F,a)\mid a\in A\}$ or $\cW=\{(F^c,a)\mid a\in A\}$. Since
  the partitions $(F,a)|R$ and $(F^c,a)|R$ are equal for any set $R$
  the proposition is true.

  Assume now that $(\cP,\cW)$ is not the initial configuration. Let
  $(\widehat\cP,\widehat\cW)$ be a configuration that precedes
  immediately $(\cP,\cW)$. Thus $(\cP,\cW)$ is obtained from
  $(\widehat\cP,\widehat\cW)$ in one step of Hopcroft's algorithm, by
  choosing one splitter $S$ in
  $\widehat\cW$, and by performing the required operations on
  $\widehat\cP$ and $\widehat\cW$.

  First we observe that, by the algorithm, and by
  Lemma~\ref{BBC:le:magique}, one has for any set of
  states $R$,
  \begin{equation}
    \label{BBC:eq:splitter}
    \bigwedge_{(\widehat W,a)\in\widehat\cW\setminus\{S\}}(\widehat
    W,a)| R\ge \bigwedge_{(W,a)\in\cW}(W,a)| R\,.
  \end{equation}
  Indeed, the set $\cW$ contains all $\widehat\cW\setminus\{S\}$ with
  the exception of those pairs $(P,a)$ for which $P$ is split into two
  parts, and in this case, the relation follows from
  \eqref{BBC:eq:magique1}.

Next, consider a subset $R$  of a set of  $\cP$.
Since $R$ was not split by $S$, that is since  $S|R=\{R\}$,
we have
\begin{equation}
  \label{BBC:eq:nosplitter}
  \bigwedge_{(\widehat W,a)\in\widehat\cW}(\widehat
    W,a)| R= \bigwedge_{(\widehat W,a)\in\widehat\cW\setminus\{S\}}(\widehat
    W,a)| R\,.
\end{equation}

Moreover, by the induction hypothesis, and since $R$ is also a subset
of a set of $\widehat \cP$, we have for any $\widehat
P\in\widehat\cP$,
\begin{equation}
  \label{BBC:eq:induction}
  (\widehat P,a)| R \ge \bigwedge_{(\widehat W,a)\in\widehat\cW}(\widehat
    W,a)| R\,.
\end{equation}
Consequently, for any subset $R$ of a set of $\cP$, any $\widehat
P\in\widehat\cP$, in view of \ref{BBC:eq:splitter},
\ref{BBC:eq:nosplitter} and \ref{BBC:eq:induction}, we have
\begin{equation}
  \label{BBC:eq:nonosplit}
 (\widehat P,a)| R \ge\bigwedge_{(W,a)\in\cW}(W,a)| R\,.
\end{equation}

Let now $R\in\cP,$ $a\in A$, and let again $R$ be a subset of a set of
$\cP.$ We consider the partition $(P,a)|R$.  We distinguish two cases.

  Case 1. Assume $P\in\widehat\cP$. The proof follows directly from
  \eqref{BBC:eq:nonosplit}. 
  
  Case 2. Assume $P\notin\widehat\cP$.  Then there exists $\widehat
  P\in\widehat\cP$ such that $S|\widehat P=\{P,P'\}$.

 If $(\widehat P,a)\in\widehat\cW\setminus\{S\}$, then
  both $(P,a)$ and $(P',a)$ are in $\cW$. Consequently
  \begin{displaymath}
    (P,a)| R\ge\bigwedge_{(W,a)\in \cW}(W,a)| R\,.
  \end{displaymath}

  If, on the contrary, $(\widehat P,a)\notin\widehat\cW\setminus\{S\}$,
  then by the algorithm, one of $(P,a)$ or $(P',a)$ is in $\cW$.

  Case 2a. Assume that $(P,a)\in\cW$. Then obviously
  \begin{displaymath}
    (P,a)| R\ge\bigwedge_{(W,a)\in \cW}(W,a)| R\,.
  \end{displaymath}
  \indent Case 2b. Assume $(P',a)\in\cW$.
By  Lemma~\ref{BBC:le:magique},
we have 
\begin{displaymath}
  (P,a)| R \ge (P',a)| R \wedge (\widehat P,a)|R
\end{displaymath}
as obviously we have 
\begin{displaymath}
    (P',a)| R\ge\bigwedge_{(W,a)\in \cW}(W,a)| R\,.
  \end{displaymath}
and by use of \eqref{BBC:eq:nonosplit},
we obtain
\begin{displaymath}
    (P,a)| R\ge\bigwedge_{(W,a)\in \cW}(W,a)| R\,.
  \end{displaymath}
 This completes the proof.
\end{proof}

\begin{corollary}
  The current partition at the end of an execution of Hopcroft's
  algorithm on an automaton $\A$ is the Nerode partition of $\A$.
\end{corollary}
\begin{proof}
  Let $\cP$ be the partition obtained at the end of an execution of
  Hopcroft's algorithm on an automaton $\A$.  By
  Proposition~\ref{BBC:prop:Nerode}, it suffices to check that no splitter
  splits a class of $\cP$. Since the waiting set $\cW$ is empty, the
  right-hand side of \eqref{BBC:eq:Hopcroftsplitter} evaluates to $\{R\}$
  for each triple $(P,a,R)$. This means that $(P,a)$ indeed does not
  split $R$.
\end{proof}

\subsection{Complexity}

\begin{proposition}\label{BBC:prop:time}
  Hopcroft's algorithm can be implemented to have worst-case time
  complexity $O(kn\log n)$ for an automaton with $n$ states over a
  $k$-letter alphabet.
\end{proposition}

To achieve the bound claimed, a partition $\cP$ of a set $Q$ should be
implemented in a way to allow the following operations:
\begin{itemize}
\item accessing the class to which  a state belongs in constant time;
\item enumeration of the elements of a class in time proportional to
  its size;
\item adding and removing of an element in a class in constant time.
\end{itemize}

The computation of all splittings $P',P''$ of classes $P$ by a given
splitter $(W,a)$ is done in time $O(\Card(a^{-1}W))$ as follows.
\begin{enumerate}
\item One enumerates the states $q$ in $a^{-1}W$. For each state $q$,
  the class $P$ of $q$ is marked as a candidate for splitting, the
  state $q$ is added to a list of states to be removed from $P$, and a
  counter for the number of states in the list is
  incremented.

\item Each class that is marked is a candidate for splitting. It is
  split if the number of states to be removed differs from the size of
  the class. If this holds, the states in the list of $P$ are removed
  to build a new class. The other states remain in $P$.
\end{enumerate}

The waiting set $\cW$ is implemented such that membership can be
tested in constant time, and splitters can be added and removed in
constant time. This allows the replacement of a splitter
$(P,b)$ by the two splitters $(P',b)$ and $(P'',b)$ in constant time,
since in fact $P'$ is just the modified class $P$, and it suffices to
add the splitter $(P'',b)$.

Several implementations of partitions that satisfy the time
requirements exist. Hopcroft \cite{Hopcroft:1971} describes such a data
structure, reported in \cite{Berstel&Perrin:2005}. Knuutila
\cite{Knuutila:2001} gives a different implementation.

\begin{proof}[Proof \textup{of Proposition~\ref{BBC:prop:time}}]
  For a given state $q$, a splitter
  $(W,a)$ such that $q\in W$ is called a \textit{$q$-splitter}.

  Consider some $q$-splitter. When it is removed from the waiting set $\cW$,
  it may be smaller than when it was added, because it may have been
  split during its stay in the waiting set. On the contrary, when a
  $q$-splitter is added to the waiting set, then its size is at most
  one half of the size it had when it was previously removed.  Thus,
  for a fixed state $q$, the number of $q$-splitters $(W,a)$ which are
  removed from $\cW$ is at most $k\log n$, since at each removal, the
  number of states in $W$ is at most one half of the previous addition.

  The total number of elements of the sets $a^{-1}W$, where $(W,a)$ is
  in $\cW$, is $O(kn\log n)$. Indeed, for a fixed state $q\in Q$, a
  state $p$ such that $p\cdot a=q$ is exactly in those sets $a^{-1}W$,
  where $(W,a)$ is a $q$-splitter in $\cW$. There are at most $O(\log n)$
  of such sets for each letter $a$, so at most $(k\log n)$ sets for
  each fixed $q$. Since there are $n$ states, the claim follows. This
  completes the proof since the running time is bounded by the size of the
  sets $a^{-1}W$, where $(W,a)$ is in $\cW$.
\end{proof}

\subsection{Miscellaneous remarks}

There are some degrees of freedom in Hopcroft's algorithm. In
particular, the way the waiting set is represented may influence the
efficiency of the algorithm. This issue has been considered in
\cite{Knuutila:2001}. In \cite{Baclet&Pagetti:2006} some practical
experiments are reported.  In \cite{Paun&Paun&Rodriguez&Paton:2009} it is
shown that the worst-case reported in \cite{Berstel&Carton:2004} in the
case of de Bruijn words remains of this complexity when the waiting
set is implemented as a queue (LIFO), whereas this complexity is never
reached with an implementation as a stack (FIFO). See this paper for
other discussions, in particular in relation with cover automata.

Hopcroft's algorithm, as reported here, requires the automaton to be
complete. This may be a serious drawback in the case where the
automaton has only a few transitions. For instance, if a dictionary is
represented by a trie, then the average number of edges per state is
about $1.5$ for the French dictionary (personal communication of
Dominique Revuz), see Table~\ref{BBC:table:size-performances} below.
Recently, two generalizations of Hopcroft's algorithm to incomplete
automata were presented, by \cite{Beal&Crochemore:2008} and
\cite{Valmari&Lehtinen:2008}, with running time $O(m\log n)$, where
$n$ is the number of states and $m$ is the number of transitions.
Since $m\le kn$ where $k$ is the size of the alphabet, the algorithm
achieves the same time complexity.

\subsection{Moore versus Hopcroft}

We present an example which illustrates the fact that Hopcroft's
algorithm is not just a refinement of Moore's algorithm. This is
proved by checking that one of the partitions computed by Moore's
algorithm in the example does not appear as a partition in any of the
executions of Hopcroft's algorithm on this automaton.

The automaton we consider is over the alphabet $A=\{a,b\}$. Its set of
states is $Q=\{0,1,2,\ldots, 9\}$, the set of final states is
$F=\{7,8,9\}$. The next-state function and the graph are given in
Figure~\ref{BBC:tab:a}. 
\begin{figure}
  \centering
\begin{picture}(60,40)(0,-20)
  \small$\begin{array}{r|cccccccccc}
   &0&1&2&3&4&5&6&7&8&9\\\hline
  a&2&1&3&3&5&7&8&1&7&6\\
  b&9&1&5&7&8&7&8&1&7&4
  \end{array}$
\end{picture}\qquad
\scalebox{0.85}{\unitlength=0.8mm\gasset{Nw=6,Nh=6}
   \begin{picture}(70,75)(0,-20)
      \node[Nmarks=i,iangle=90](0)(0,20){$0$}
      \node(1)(60,0){$1$}
      \node(2)(20,30){$2$}
      \node(3)(40,40){$3$}
      \node(4)(20,10){$4$}
      \node(5)(40,20){$5$}
      \node(6)(20,-10){$6$}
      \node[Nmarks=f,fangle=0](7)(60,20){$7$}
      \node[Nmarks=f,fangle=-90](8)(40,0){$8$}
      \node[Nmarks=f,fangle=-90](9)(0,0){$9$}
      \drawedge(0,2){$a$}\drawedge(0,9){$b$}
      \drawloop[loopangle=-90](1){$a,b$}
      \drawedge(2,3){$a$} \drawedge[ELside=r](2,5){$b$}
      \drawloop[loopangle=90](3){$a$}\drawedge(3,7){$b$}
      \drawedge(4,5){$a$} \drawedge[ELside=r](4,8){$b$}
      \drawedge(5,7){$a,b$}
      \drawedge[ELside=r](6,8){$a,b$}
      \drawedge(7,1){$a,b$}
      \drawedge[ELside=r](8,7){$a,b$}
      \drawedge[ELside=r](9,6){$a$}\drawedge(9,4){$b$}
   \end{picture}}      
  \caption{Next-state function of the example automaton.}
  \label{BBC:tab:a}
\end{figure}
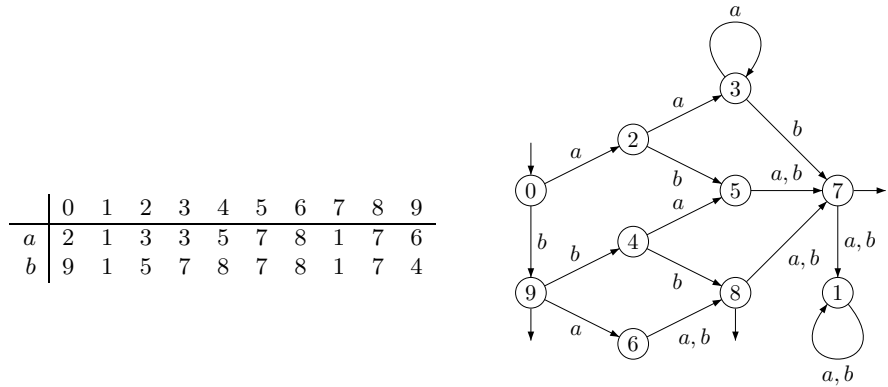

The Moore partitions are easily computed. The partition $\cM_1$ of
order~$1$ is composed of the five classes:
\begin{displaymath}
  \{0,3,4\},\{1,2\},\{5,6\},\{7,9\},\{8\}\,.
\end{displaymath}
The Moore partition of order~$2$ is the identity.

The initial partition for Hopcroft's algorithm is $\{F,F^c\}$, where
$F=789$ (we will represent a set of states by the sequence of its
elements). The initial waiting set is composed of $(789,a)$ and
$(789,b)$.  There are two cases, according to the choice of the first
splitter.

Case 1. The first splitter is $(789,a)$. Since $a^{-1}789=568$, each
of the classes $F$ and $F^c$ is split. The new partition is
$01234|56|79|8$. The new waiting set is
  \begin{displaymath}
    (79,b), (8,b),(8,a),(56,a),(56,b)\,.
  \end{displaymath}
The first three columns in 
Table~\ref{BBC:tab:m} contain the sets $c^{-1}P$, for $(P,c)$ in this
waiting set. By inspection, one sees that each entry in these columns
cuts off at least one singleton class which is not in the Moore
equivalence $\cM_1$. This implies that $\cM_1$ cannot be obtained by
Hopcroft's algorithm in this case.

Case 2. The first splitter is $(789,b)$. Since $b^{-1}789=034568$, the
new partition is $12|03456|79|8$. The new waiting set is
\begin{displaymath}
  (79,a),(8,a),(8,b),(12,a),(12,b)\,.
\end{displaymath}
Again, each entry in the last three columns of Table~\ref{BBC:tab:m}
cuts off at least one singleton class which is not in the Moore
equivalence $\cM_1$. This implies that, also in this case, $\cM_1$
cannot be obtained by Hopcroft's algorithm.

\begin{table}
  \centering
  $\begin{array}{r|cccc}
    P&56&8&79&12\\\hline
    \text{\Large\vphantom(} a^{-1}P%
    &49&6&58&017\\
    \text{\Large\vphantom(} b^{-1}P%
    &2&46&0357&17
  \end{array}$
  \caption{The sets $c^{-1}P$, with $(P,c)$ in a waiting set.}
  \label{BBC:tab:m}
\end{table}

\bigskip

Despite the difference illustrated by
this example, there are similarities between Moo\-re's and
Hopcroft's algorithms that have been exploited by Julien David in his
thesis~\cite{David:2010b} to give an upper bound on the average
running time of Hopcroft's algorithm for a particular
strategy. 

In this strategy,
there are two waiting sets, the current set $\cW$ and a \emph{future
  waiting set} $\cF$. Initially, $\cF$ is empty. Hopcroft's algorithm
works as usual, except for line $14$: Here, the splitter
$(\min(P',P''),b)$ is added to $\cF$ and not to $\cW$. When $\cW$ is
empty, then the contents of $\cF$ and $\cW$ are swapped. The algorithm
stops when both sets $\cW$ and $\cF$ are empty.

\begin{proposition}[David \cite{David:2010b}]
  There is a strategy for Hopcroft's algorithm such that its average
  complexity, for the uniform probability over all complete automata
  with $n$ states, is $O(n\log\log n)$.
\end{proposition}

Julien David shows that at the end of each cycle, that is when $\cW$
becomes empty, the current partition $\cP$ of the set of states is in
fact a refinement of the corresponding level in Moore's algorithm. This
shows that the number of cycles in Hopcroft's algorithm, for this
strategy, is bounded by the depth of the automaton. Thus
Theorem~\ref{BBC:theoremDavid} applies.

%
%

\section{Slow automata} \label{BBC:sec:slow}

We are concerned in this section with automata that behave badly for
Hopcroft's and Moore's minimization algorithms. In other terms, we
look for automata for which Moo\-re's algorithm requires the maximal
number of steps, and similarly for Hopcroft's algorithm.

\subsection{Definition and equivalence}

Recall that an automaton with $n$ states is called \emph{slow for
  Moore}\index{automaton!slow for Moore}\index{slow automaton!for
  Moore} if the number $\ell$ of steps in Moore's algorithm is
$n-2$. A slow automaton is minimal.
It is equivalent to say that each Moore equivalence $\equiv_h$ has
exactly $h+2$ equivalence classes for $h\le n-2$. This is due to the
fact that, at each step, just one class of $\equiv_h$ is split, and
that this class is split into exactly two classes of the equivalence
$\equiv_{h+1}$.

Proposition~\ref{BBC:prop:Moore} takes the following special form for slow
automata. 
\begin{proposition} \label{BBC:cor:Mooreslow}
  Let $\A$ be an automaton with $n$ states which is slow for Moore.
  For all $n-2>h\ge0$, there is exactly one class $R$ in $\cM_h$ which
  is split, and moreover, if $(P,a)$ and $(P',a')$ split $R$, with
  $P,P'\in\cM_h$, then $(P,a)|R=(P',a')|R$.\qed
\end{proposition}

An automaton is \emph{slow for Hopcroft}\index{automaton!slow for
  Hopcroft}\index{slow automaton!for Hopcroft} if, for \emph{all
  executions} of Hopcroft's algorithm, the splitters in the current
waiting set either do not split or split in the same way: there is a
unique class that is split into two classes, and always into the same
two classes.

More formally, at each step $(\cW,\cP)$ of an execution, there is
at most one class $R$ in the current partition $\cP$ that is split, and
for all splitters $(P,a)$ and $(P',a')$ in $\cW$ that split $R$, one has
$(P,a)|R=(P',a')|R$. 

The definition is close to the statement in
Proposition~\ref{BBC:cor:Mooreslow} above, and indeed, one has the
following property.

\begin{theorem}
  An automaton is slow for Moore if and only if it is slow for Hopcroft.
\end{theorem}

\begin{proof} Let $\A$ be a finite automaton.  We first suppose that
  $\A$ is slow for Moore.  We consider an execution of Hopcroft's
  algorithm, and we prove that each step of the execution that changes
  the partition produces a Moore partition.

  This holds for the initial configuration $(\cW,\cP)$, since
  $\cP=\cM_0$. Assume that one has $\cP=\cM_h$ for some configuration
  $(\cW,\cP)$ and some $h\ge0$.  Let $R$ be the class of $\cM_h$ split
  by Moore's algorithm.

  Let $S\in\cW$ be the splitter chosen in Hopcroft's algorithm. Then
  either $S$ splits no class, and the partition remains equal to
  $\cM_h$ or by Proposition~\ref{BBC:cor:Mooreslow} it splits the class $R$.
  In the second case, this class is split by $S$ into two new classes,
  say $R'$ and $R''$. The partition
  $\cP'=\cP\setminus\{R\}\cup\{R',R''\}$ is equal to $\cM_{h+1}$.

  Conversely, suppose that $\A$ is slow for Hopcroft. We show that it
  is also slow for Moore by showing that the partition $\cM _ {h+1}$
  has only one class more than $\cM_h$. For this, we use
  Proposition~\ref{BBC:prop:Moore} which states that each class $R$ in
  $\cM_h$ is refined in $\cM _ {h+1}$ into the partition $\mathcal{R}$
  given by
  \begin{displaymath}
    \mathcal{R}=\bigwedge_{a\in A}\bigwedge_{P\in\cM_{h}}(P,a)|R\,.
  \end{displaymath}
  We show by induction on the number of steps that, in any execution
  of Hopcroft's algorithm, $\cP=\cM_h$ for some configuration
  $(\cP,\cW)$.  This holds for the initial configuration.  Let $(W,a)$
  be some splitter in $\cW $. It follows from
  Proposition~\ref{BBC:prop:Hopcroftsplitter} that
  \begin{displaymath}
    \mathcal{R}\ge \bigwedge_{a\in A}\bigwedge_{(W,a)\in \cW}(W,a)|R\,.
  \end{displaymath}
  Thus the partition $\mathcal{R}$ is coarser than that the partition
  of $R$ obtained by Hopcroft's algorithm. Since the automaton is slow
  for Hopcroft, the partition on the right-hand side has at most two
  elements. More precisely, there is exactly one class $R$ that is
  split by Hopcroft's algorithm into two classes. Since Moore's
  partition $\cM _ {h+1}$ is coarser, it contains precisely these
  classes. This proves the claim.
\end{proof}

%
%
\subsection{Examples}

  \begin{figure}
    \centering 
    \scalebox{0.85}{\gasset{Nw=6,Nh=6}
    \begin{picture}(60,18)(0,-8)
      \node[Nmarks=i,iangle=180](1)(0,0){$0$}
      \node(2)(15,0){$1$}
      \node(3)(30,0){$2$}
      \node[Nframe=n,Nw=8](d)(45,0){$\cdots$}
      \node[Nmarks=f,fangle=0](n)(60,0){$n$}
      \drawloop(n){$a$}
      \drawedge(1,2){$a$}
      \drawedge(2,3){$a$}
      \drawedge(3,d){$a$}
      \drawedge(d,n){$a$}
    \end{picture}}  
    \caption{An automaton over one letter recognizing the set of words
      of length at least $n$.}
    \label{BBC:fig:slow1}
  \end{figure}
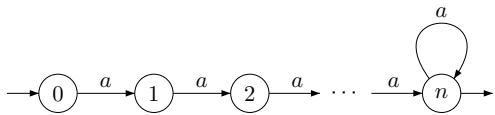

\begin{example}
  The simplest example is perhaps the automaton given in
  Figure~\ref{BBC:fig:slow1}.  For $0\le h\le n-1$, the partition
  $\cM_h$ is composed of the class $\{0,\ldots,n-h-1\}$, and of the
  singleton classes $\{n-h\}$, $\{n-h+1\}$,\dots, $\{n\}$. At each
  step, the last state is split off from the class
  $\{0,\ldots,n-h-1\}$.
\end{example}

\begin{example}
  The automaton of Figure~\ref{BBC:fig:slow2} recognizes the set
  $D^{(n)}$ of Dyck
  words $w$ over $\{a,b\}$ such that $0\le|u|_a-|u|_b\le n$ for all
  prefixes $u$ of $w$. The partition $\cM_h$, for $0\le
  h\le n-1$, is composed of $\{0\}, \ldots, \{h\}$, and
  $\{\infty,h+1,\dots,n\}$. At $h=n$, the state $\{\infty\}$ is
  separated from state $n$.
\end{example}

  \begin{figure}
    \centering \scalebox{0.85}{\gasset{Nw=6,Nh=6,curvedepth=4}
   \begin{picture}(75,18)(-15,-8)
      \node(x)(-15,0){$\infty$}
      \node[Nmarks=if,iangle=90,fangle=90](1)(0,0){$0$}
      \node(2)(15,0){$1$}
      \node(3)(30,0){$2$}
      \node[Nframe=n,Nw=8](d)(45,0){$\cdots$}
      \node(n)(60,0){$n$}
      \drawedge[curvedepth=0](1,x){$b$}
      \drawedge(2,1){$b$}
      \drawedge(3,2){$b$}
      \drawedge(d,3){$b$}
      \drawedge(n,d){$b$}
      \drawedge(1,2){$a$}
      \drawedge(2,3){$a$}
      \drawedge(3,d){$a$}
      \drawedge(d,n){$a$}
   \end{picture}}   
    \caption{An automaton recognizing the Dyck words of ``height'' at most
      $n$.}
    \label{BBC:fig:slow2}
  \end{figure}
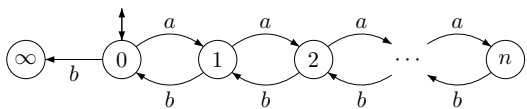

\begin{example}\label{BBC:example:slow}
  Let $w=b_1\cdots b_n$ be a word of length~$n$ over the binary
  alphabet $\{0,1\}$.  We define an automaton $\mathcal{A}_w$ over the
  unary alphabet $\{a\}$ as follows.  The state set of $\mathcal{A}_w$
  is $\{1,\ldots,n\}$ and the next state function is defined by $i
  \cdot a = i+1$ for $i < n$ and $n \cdot a = 1$.  Note that the
  underlying labeled graph of~$\mathcal{A}_w$ is just a cycle of
  length~$n$.  The final states really depend on~$w$.  The set of
  final states of~$\mathcal{A}_w$ is $F = \{ 1 \le i \le n \mid b_i =
  1\}$. We call such an automaton a \emph{cyclic
    automaton}\index{cyclic automaton}\index{automaton!cyclic}.
  \begin{figure}
    \begin{center}
      \scalebox{0.85}{%
      \gasset{Nw=6,Nh=6,curvedepth=1.5}
      \begin{picture}(35,35)(-15,-15) \node[Nmarks=i](1)(-15,0){$1$}
        \node[Nmarks=r](2)(-10.6,10.6){$2$} \node(3)(0,15){$3$}
        \node(4)(10.6,10.6){$4$} \node[Nmarks=r](5)(15,0){$5$}
        \node(6)(10.6,-10.6){$6$} \node[Nmarks=r](7)(0,-15){$7$}
        \node(8)(-10.6,-10.6){$8$} \drawedge(1,2){$a$}
        \drawedge(2,3){$a$} \drawedge(3,4){$a$} \drawedge(4,5){$a$}
        \drawedge(5,6){$a$} \drawedge(6,7){$a$} \drawedge(7,8){$a$}
        \drawedge(8,1){$a$}
      \end{picture}}
      \caption{Cyclic automaton $\mathcal{A}_w$ for $w =
        01001010$. Final states are circled.}
      \label{BBC:fig:11101000}
    \end{center}
  \end{figure}
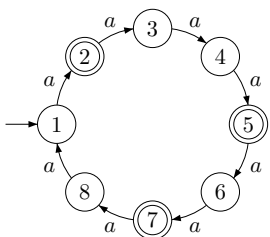

  For a binary word $u$, we define $Q_u$ to be the set of states of
  $\mathcal{A}_w$ which are the starting positions of circular
  occurrences of~$u$ in~$w$.  If $u$ is the empty word, then $Q_u$ is
  by convention the set~$Q$ of all states of~$\mathcal{A}_w$.  By
  definition, the set~$F$ of final states of~$\mathcal{A}_w$ is $Q_1$
  while its complement $F^c$ is~$Q_0$.

  Consider the automaton $\mathcal{A}_w$ for $w = 01001010$ given in
  Fig.~\ref{BBC:fig:11101000}.  The sets $Q_1$, $Q_{01}$ and $Q_{11}$ of
  states are respectively $\{2,5,7\}$, $\{1,4,6\}$ and $\emptyset$. If
  $w$ is a Sturmian word, then the automaton $\mathcal{A}_w$ is slow.
\end{example}

Slow automata are closely related to binary Sturmian trees. Consider
indeed a finite automaton, for instance over a binary alphabet. To
this automaton corresponds an infinite binary tree, composed of all
paths in the automaton. The nodes of the tree are labeled with the
states encountered on the path.  For each integer $h\ge0$, the number
of distinct subtrees of height $h$ is equal to the number of classes
in the Moore partition $\cM_h$.  It follows that the automaton is slow
if and only if there are $h+1$ distinct subtrees of height $h$ for all
$h$: this is precisely the definition of \emph{Sturmian
  trees}\index{Sturmian!tree}\index{tree!Sturmian}, as given in
\cite{Berstel&Boasson&Carton&Fagnot:2010}.

We consider now the problem of showing that the running time $O(n\log
n)$ for Hopcroft's algorithm on $n$-state automata is tight. The
algorithm has a degree of freedom because, in each step of its main
loop, it allows one to choose the splitter to be processed. Berstel and
Carton~\cite{Berstel&Carton:2004} introduced a family of finite
automata based on de Bruijn words. These are exactly the cyclic
automata $\A_w$ of Example~\ref{BBC:example:slow} where $w$ is a
binary de Bruijn word. They showed that there exist some ``unlucky''
sequence of choices that slows down the computation to achieve the lower
bound $\Omega(n\log n)$.

In the papers~\cite{Castiglione&Restivo&Sciortino:2007}
and~\cite{Castiglione&Restivo&Sciortino:2008}, Castiglione, Restivo
and Sciortino replace de Bruijn words by Fibonacci words. They observe
that for these words, and more generally for all circular standard
Sturmian words, there is no more choice in Hopcroft's
algorithm. Indeed, the waiting set always contains only one
element. The uniqueness of the execution of Hopcroft's algorithm
implies by definition that the associated cyclic automata for Sturmian
words are slow.

They show that, for Fibonacci words, the unique execution of
Hopcroft's algorithm runs in time $\Omega(n\log n)$, so that the
worst-case behavior is achieved for the cyclic automata of Fibonacci
words. The computation is carried out explicitly, using connections
between Fibonacci numbers and Lucas
numbers. In~\cite{Castiglione&Restivo&Sciortino:2009}, they give a
detailed analysis of the reduction process that is the basis of their
computation, and they show that this process is isomorphic, for all
standard Sturmian words, to the refinement process in Hopcroft's
algorithm.

In~\cite{Berstel&Boasson&Carton:2009b}, the analysis of the running
time of Hopcroft's algorithm is extended to cyclic automata of
standard Sturmian words. It is shown that the directive sequences for
which Hopcroft's algorithm has worst-case running time are those
sequences $(d_1,d_2,d_3,\ldots)$ for which the sequence of geometric
means $((p_n)^{1/n})_{n\ge1}$, where $p_n=d_1d_2\cdots d_n$, is bounded.



%
%
\section{Minimization by fusion}\label{BBC:sec:fusion}

In this section, we consider the minimization of automata by fusion of
states. An important application of this method is the computation of
the minimal automaton recognizing a given finite set of words. This is
widely used in computational linguistics for the space-efficient
representation of dictionaries.


Let $\A$ be a deterministic automaton over the alphabet $A$, with set
of states $Q$.  The \emph{signature}\index{signature of a
  state}\index{state!signature} of a state $p$ is the set of pairs
$(a,q)\in A\times Q$ such that $p\cdot a=q$, together with a Boolean
value denoting whether $p$ is final or not.  Two states $p$ and
$q$ are called \emph{mergeable}\index{mergeable
  states}\index{state!mergeable} if and only if they have the same
signature.  The \textit{fusion}\index{fusion of
  states}\index{state!fusion} or \textit{merge}\index{merge of
  states}\index{state!merge} of two mergeable states $p$ and $q$
consists in replacing $p$ and $q$ by a single state.  The state
obtained by fusion of two mergeable states has the same signature.

Minimization of an automaton by a sequence of fusion of states with
the same signature is not always possible. Consider the two-state
automaton over the single letter $a$ given in
Figure~\ref{BBC:fig:nofusion} which recognizes $a^*$. It is not minimal. The
signature of state $1$ is $+,(a,2)$ and the signature of state $2$ is
$+,(a,1)$ (here ``$+$'' denotes an accepting state), so the states have
different signatures and are not mergeable.

\begin{figure} [h]
  \centering
  \scalebox{0.85}{\gasset{Nw=6,Nh=6,curvedepth=5}
  \begin{picture}(24,12)(0,-6)
    \node[Nmarks=if,iangle=90,fangle=90](1)(0,0){$1$}
    \node[Nmarks=f,fangle=90](2)(24,0){$2$}
    \drawedge(1,2){$a$}
    \drawedge(2,1){$a$}
  \end{picture}}   
  \caption{An automaton recognizing the set $a^*$ which can not be
    minimized by fusion of its states.}
  \label{BBC:fig:nofusion}
\end{figure}
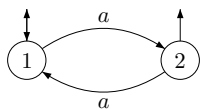

\subsection{Local automata}

M.-P. B\'eal and M. Crochemore
\cite{Beal&Crochemore:2007} designed an
algorithm for minimizing a special class of deterministic automata by
a sequence of state mergings. These automata are called irreducible
\emph{local} automata. They occur quite naturally in symbolic dynamics. The
running time of the algorithm is $O(\min(m(n-r+1),m\log n))$, where
$m$ is the number of edges, $n$ is the number of states, and $r$ is
the number of states of the minimized automaton. In particular, the
algorithm is linear when the automaton is already minimal. Hopcroft's
algorithm has running time $O(kn\log n)$, where $k$ is the size of the
alphabet, and since $kn\ge m$, it is worse than B\'eal and
Crochemore's algorithm. Moreover, their algorithm does not require the
automaton to be complete.

The automata considered here have several particular features. First,
all states are both initial and final. Next, they are
\emph{irreducible}, that is, their underlying graph is strongly
connected. Finally, the automata are \emph{local}\index{local
  automaton}\index{automaton!local}. By definition, this means that
two distinct cycles carry different labels. This implies that the
labels of a cycle are primitive words, since otherwise there exist
different traversals of the cycle which have the same label.  In
\cite{Beal&Crochemore:2007}, the constructions and proofs are done for
a more general family of automata called AFT (for automata of
\emph{almost finite type}\index{automaton!almost finite type}). We
sketch here the easier case of local automata.

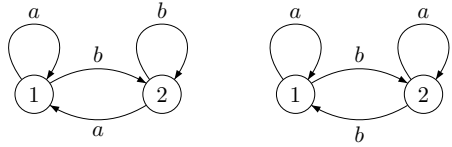
\begin{figure}
  \centering\scalebox{0.85}{
  \gasset{Nh=6,Nw=6}
  \begin{picture}(20,20)(0,-5)
    \node(1)(0,0){$1$}
    \node(2)(20,0){$2$}
    \drawloop(1){$a$}\drawloop(2){$b$}
    \drawedge[curvedepth=4](1,2){$b$}
    \drawedge[curvedepth=4](2,1){$a$}
  \end{picture}\hspace*{2cm}
  \begin{picture}(20,20)(0,-5)
    \node(1)(0,0){$1$}
    \node(2)(20,0){$2$}
    \drawloop(1){$a$}\drawloop(2){$a$}
    \drawedge[curvedepth=4](1,2){$b$}
    \drawedge[curvedepth=4](2,1){$b$}
  \end{picture}}
 \caption{The automaton on the left is local, the automaton on the
   right is not because of the two loops labeled $a$, and because the
   label $bb$ of the cycle through $1$ and $2$ is not a primitive word.}
   \label{BBC:fig:localnonlocal}
\end{figure}

\noindent Since all states are final, two states $p$ and $q$ of an
automaton are \emph{mergeable} if and only if, for all letters $a\in
A$, $p\cdot a$ is defined if and only if $q\cdot a$ and, if this is
the case, then $p\cdot a=q\cdot a$.
 
The basic proposition is the following. It shows that an irreducible
local automaton can be minimized by a sequence of fusion of states.

\begin{proposition}
  If an irreducible local automaton is not minimal, then at least two
  of its states are mergeable.\qed
\end{proposition}

The minimization algorithm assumes that the alphabet is totally
ordered. It uses the notion of partial signature. First, with each
state $q$ is associated the signature $\sigma(q)=a_1p_1a_2p_2\cdots
a_mp_m$, where $\{(a_1,p_1),\ldots,(a_m,p_m)\}$ is the signature of
$q$, and the sequence is ordered by increasing value of the letters.
Since all states are final, the Boolean indicator reporting this is
omitted. A \emph{partial signature}\index{partial
  signature}\index{state!partial signature}\index{signature!partial}
is any prefix $a_1p_1a_2p_2\cdots a_ip_i$ of a signature.

A first step consists in building a \emph{signature
  tree}\index{tree!signature}\index{signature!tree} which 
represents the sets of states sharing a common partial signature. The
root of the tree represents the set of all states, associated to the
empty signature. A node representing the sets of states with a partial
signature $a_1p_1a_2p_2\cdots a_ip_i$ is the parent of the nodes
representing the sets of states with a partial signature
$a_1p_1a_2p_2\cdots a_ip_ia_{i+1}p_{i+1}$. As a consequence, leaves
represent full signatures. All states that correspond to a leaf are
mergeable.

When mergeable states are detected in the signature tree, they can be
merged. Then the signature tree has to be updated, and this is the
difficult part of the algorithm.

\subsection{Bottom-up  minimization}

In this section, all automata are finite, acyclic, deterministic and
trim.  
A state $p$ is called \textit{confluent}\index{confluent
  states}\index{state!confluent} if there are at least two edges in
$\A$ ending in $p$.

A \textit{trie}\index{trie}\index{automaton!trie} is an automaton
whose underlying graph is a tree.  Thus an automaton is a trie if and
only if it has no confluent state.

\textit{Bottom-up minimization} is the process of minimizing an
acyclic automaton by a bot\-tom-up traversal. In such a traversal,
children of a node are treated before the node itself. During the
traversal, equivalent states are detected and merged. The basic
property of bottom-up minimization is that the check for (Nerode)
equivalence reduces to equality of signatures. The critical point is to
organize the states that are candidates in order to do this check
efficiently.

The bottom-up traversal itself may be organized in several ways, for
instance as a depth-first search with the order of traversal of the
children determined by the order on the labels of the edges. Another
traversal is by increasing height, as done in Revuz's algorithm given
next.

One popular method for the construction of a minimal automaton for a
given finite set of words consists in first building a trie for this
set and then minimizing it. Daciuk \textit{et
  al.}~\cite{Daciuk&Mihov&Watson&Watson:2000} propose an incremental
version which avoids this intermediate construction.

Recall that the \textit{signature}\index{signature of a
  state}\index{state!signature} of a state $p$ is the set of pairs
$(a,q)$ such that $p\cdot a=q$ together with a Boolean value
indicating whether the state is final or not.  It is tacitly
understood that the alphabet of the automaton is ordered.  The
signature of a state is usually considered as the ordered sequence of
pairs, where the order is determined by the letters. It is important
to observe that the signature of a state evolves when states are
merged. As an example, the state $6$ of the automaton on the left of
Figure~\ref{BBC:fig:AZ1} has signature $+,(a,3),(b,10)$, and the same
state has signature $+,(a,3),(b,7)$ in the automaton on the right of
the figure.

As already mentioned, if the minimization of
the children of two states $p$ and $q$ has been done, then $p$ and $q$
are (Nerode) equivalent if and only if they have the same signature.
So the problem to be considered is the bookkeeping of signatures, that
is the problem of detecting whether the signature of the currently
considered state has already occured before. In practical
implementations, this is done by hash coding the signatures. This
allows one to perform the test in constant average time. One
remarkable exception is Revuz's algorithm to be presented now, and its
extension by Almeida and Zeitoun that we describe later.

\subsection{Revuz's algorithm}

Revuz \cite{Revuz:1992} was the first to give an explicit description
of a linear time implementation of the bottom-up minimization
algorithm. The principle of the algorithm was also described by
\cite{Krivol:1991}. 

\begin{figure}
  \setlength{\commentspace}{6cm}
\hrule\smallskip\par
  \begin{algorithmic}
  \FUNC{Revuz$(\A)$}
  \FOR{$h=0$ \TO $\textsc{Height}(\A)$}
    \STATE\algcomment{-1}{Compute states of height $h$}%
    $S\gets\textsc{GetStatesForHeight}(h)$
    \STATE\algcomment{-1}{Compute and sort signatures}$\textsc{SortSignatures}(S)$ 
    \FOR{\algcomment{0.5}{Merge mergeable states}$s\in S$}
      \IF{$s$ and $s.next$ have the same signature}
        \STATE\algcomment{1}{Mergeable  states are consecutive}$\textsc{Merge}(s,s.next)$
      \ENDIF
    \ENDFOR
  \ENDFOR
  \end{algorithmic}
\hrule\smallskip
\caption{Revuz's minimization algorithm.}\label{BBC:alg:revuz}
\end{figure}

Define the \emph{height}\index{height of a state}\index{state!height}
of a state $p$ in an acyclic automaton to be the length of the longest
path starting at $p$. It is also the length of the longest word in the
language of the subautomaton at $p$. Two equivalent states have the
same height. Revuz's algorithm operates by increasing height. It is
outlined in Figure~\ref{BBC:alg:revuz}. Heights may be computed in
linear time by a bottom-up traversal. The lists of states of a given
height are collected during this traversal. The signature of a state
is easy to compute provided the edges starting in a state have been
sorted (by a bucket sort for instance to remain within the linear time
constraint).  Sorting states by their signature again is done by a
lexicographic sort.  As for Moore's algorithm, the last step can by
done by a simple scan of the list of states since states with equal
signature are consecutive in the sorted list.

The whole algorithm can be implemented to run in time $O(m)$ for an
automaton with $m$ edges.  

Revuz's algorithm relies on a particular bottom-up traversal of the
trie.  This traversal is defined by increasing height of states, and
it makes the check for equivalent states easier. With another method
for checking signatures in mind, like hash coding, the algorithm may
be organized in a different way.  For instance, the traversal by
heights can be replaced by a traversal by lexicographic order. The
last item in the algorithm may be replaced by another check.
Whenever a state has been found which must be in the minimal
automaton, its hash code is \emph{registered}. When the signature of a
state is computed, one checks whether its hash code is registered.
If not, it is added to the register, otherwise it is replaced by the
hash code.

Several implementations have been given in various packages. A
comparison has been given in \cite{Daciuk:2002}. See also
Table~\ref{BBC:table:size-performances} below for numerical data.


%
%

\subsection{The algorithm of Almeida and Zeitoun}

Almeida and Zeitoun \cite{Almeida&Zeitoun:2008} consider an extension
of the bottom-up minimization algorithm to automata which contain only
simple cycles. They describe a linear time algorithm for these
automata.

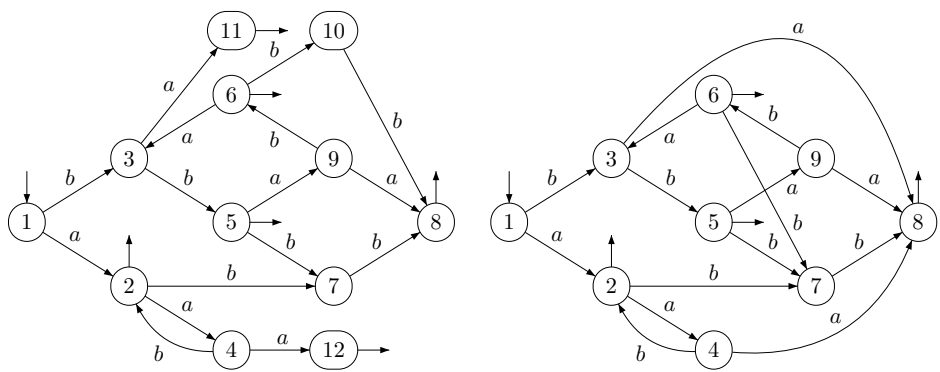
\begin{figure}
\centering\scalebox{0.85}{%
  \begin{picture}(64,60)(0,5)
    \gasset{Nh=6,Nadjust=w, Nadjustdist=2}
    \node[Nmarks=i,iangle=90](1)(0,30){$1$}
    \node[Nmarks=f,fangle=90](2)(16,20){$2$}
    \node(3)(16,40){$3$}
    \node(4)(32,10){$4$}
    \node[Nmarks=f,fangle=0](5)(32,30){$5$}
    \node[Nmarks=f,fangle=0](6)(32,50){$6$}
    \node(7)(48,20){$7$}
    \node[Nmarks=f,fangle=90](8)(64,30){$8$}
    \node(9)(48,40){$9$}
    \node(10)(48,60){$10$}
    \node[Nmarks=f,fangle=0](11)(32,60){$11$}
    \node[Nmarks=f,fangle=0](12)(48,10){$12$}
    \drawedge[curvedepth=4](4,2){$b$}
    \drawedge(4,12){$a$}
    \drawedge(5,7){$b$}
    \drawedge(5,9){$a$}
    \drawedge(6,3){$a$}
    \drawedge(6,10){$b$}
    \drawedge(7,8){$b$}
    \drawedge(9,8){$a$}
    \drawedge(9,6){$b$}
    \drawedge(10,8){$b$}
    \drawedge(3,5){$b$}
    \drawedge(3,11){$a$}
    \drawedge(2,4){$a$}
    \drawedge(2,7){$b$}
    \drawedge[ELpos=40](1,2){$a$}
    \drawedge(1,3){$b$}
  \end{picture}
\qquad\quad
 \begin{picture}(64,55)(0,5)
    \gasset{Nh=6,Nadjust=w, Nadjustdist=2}
    \node[Nmarks=i,iangle=90](1)(0,30){$1$}
    \node[Nmarks=f,fangle=90](2)(16,20){$2$}
    \node(3)(16,40){$3$}
    \node(4)(32,10){$4$}
    \node[Nmarks=f,fangle=0](5)(32,30){$5$}
    \node[Nmarks=f,fangle=0](6)(32,50){$6$}
    \node(7)(48,20){$7$}
    \node[Nmarks=f,fangle=90](8)(64,30){$8$}
    \node(9)(48,40){$9$}
    \drawedge[curvedepth=4](4,2){$b$}
    \drawedge[curvedepth=-8](4,8){$a$}
    \drawedge(5,7){$b$}
    \drawedge[ELpos=70,ELside=r](5,9){$a$}
    \drawedge(6,3){$a$}
    \drawedge[ELpos=70](6,7){$b$}
    \drawedge(7,8){$b$}
    \drawedge(9,8){$a$}
    \drawedge[ELside=r](9,6){$b$}
    \drawedge(3,5){$b$}
    \drawedge[curvedepth=24](3,8){$a$}
    \drawedge(2,4){$a$}
    \drawedge(2,7){$b$}
    \drawedge[ELpos=40](1,2){$a$}
    \drawedge(1,3){$b$}
  \end{picture}
}

\caption{On the left a simple automaton: its nontrivial strongly
  connected components are the cycles $2,4$ and $3,5,9,6$. The
  minimization starts by merging $11$, $12$ and $8$, and also $10$ and
  $7$. This gives the automaton on the right.  }
  \label{BBC:fig:AZ1}
\end{figure}

Let $\A$ be a finite trim automaton. We call it
\emph{simple}\index{simple!automaton}\index{automaton!simple} if every
nontrivial strongly connected component is a simple cycle, that is if
every vertex of the component has exactly one successor vertex in this
component. The automaton given on the left of Figure~\ref{BBC:fig:AZ1}
is simple. Simple automata are interesting because they recognize
exactly the bounded regular languages or, equivalently, the languages
with polynomial growth. These are the simplest infinite regular
languages.

The starting point of the investigation of \cite{Almeida&Zeitoun:2008}
is the observation that minimization can be split into two parts:
minimization of an acyclic automaton and minimization of the set of
strongly connected components. There are three subproblems, namely (1)
minimization of each strongly connected component, (2) identification
and fusion of isomorphic minimized strongly connected components, and (3)
wrapping, which consists in merging states which are equivalent to a
state in a strongly connected component, but which are not in this
component. The authors show that if these subproblems can be solved in
linear time, then, by a bottom-up algorithm which is a sophistication
of Revuz's algorithm, the whole automaton can be minimized in linear
time. Almeida and Zeitoun show how this can be done for simple
automata. The outline of the algorithm is given in
Figure~\ref{BBC:alg:az}.

\begin{figure}[b]
  \setlength{\commentspace}{6cm}
\hrule\smallskip\par
  \begin{algorithmic}
    \FUNC{AlmeidaZeitoun$(\A)$}
    \STATE\algcomment{0}{States and cycles of height $0$}%
    $S\gets\textsc{ZeroHeight}(\A)$ 
    \WHILE{$S\ne\emptyset$}
    \STATE\algcomment{1}{Minimize each cycle}%
    $\textsc{MinimizeCycles}(S)$
    \STATE\algcomment{1}{Compute and sort signatures and merge}%
    $\textsc{MergeIsomorphicCycles}(S)$ 
    \STATE\algcomment{1}{Search states to wrap}%
    $\textsc{Wrap}(S)$ 
    \STATE\algcomment{1}{Compute states for next height}%
    $S\gets\textsc{NextHeight}(\A,S)$ 
    \ENDWHILE
  \end{algorithmic}
\hrule\smallskip
\caption{The  algorithm of Almeida and Zeitoun.}\label{BBC:alg:az}
\end{figure}

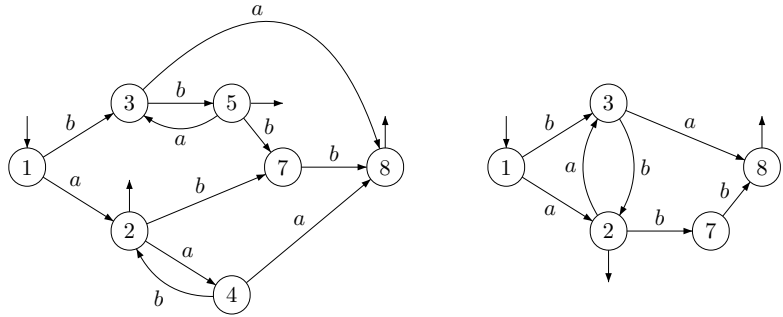
\begin{figure}
\centering\scalebox{0.85}{%
  \begin{picture}(60,55)(0,5)
    \gasset{Nh=6,Nadjust=w, Nadjustdist=2}
    \node[Nmarks=i,iangle=90](1)(0,30){$1$}
    \node[Nmarks=f,fangle=90](2)(16,20){$2$}
    \node(3)(16,40){$3$}
    \node(4)(32,10){$4$}
    \node[Nmarks=f,fangle=0](5)(32,40){$5$}
    \node(7)(40,30){$7$}
    \node[Nmarks=f,fangle=90](8)(56,30){$8$}
    \drawedge[curvedepth=18,ELpos=40](3,8){$a$}
    \drawedge(3,5){$b$}
    \drawedge(2,4){$a$}
    \drawedge(2,7){$b$}
    \drawedge[curvedepth=4](4,2){$b$}
    \drawedge(4,8){$a$}
    \drawedge[curvedepth=4](5,3){$a$}
    \drawedge(5,7){$b$}
    \drawedge(7,8){$b$}
    \drawedge[ELpos=40](1,2){$a$}\drawedge(1,3){$b$}
  \end{picture}
  \qquad\qquad
  \begin{picture}(40,30)(0,5)
    \gasset{Nh=6,Nadjust=w, Nadjustdist=2}
    \node[Nmarks=i,iangle=90](1)(0,30){$1$}
    \node[Nmarks=f,fangle=-90](2)(16,20){$2$}
    \node(3)(16,40){$3$}
    \node(7)(32,20){$7$}
    \node[Nmarks=f,fangle=90](8)(40,30){$8$}
    \drawedge(3,8){$a$}
    \drawedge(2,7){$b$}
    \drawedge[curvedepth=4](3,2){$b$}
    \drawedge[curvedepth=4](2,3){$a$}
    \drawedge(7,8){$b$}
    \drawedge[ELside=r](1,2){$a$}\drawedge(1,3){$b$}
  \end{picture}
}

\caption{The minimization continues by merging the states $5$ and $6$,
  and the states $3$ and $9$. This gives the automaton on the
  left. The last step of minimization merges the states $2$ and $5$,
  and the states $3$ and $4$.}
  \label{BBC:fig:AZ2}
\end{figure}

The algorithm works as Revuz's algorithm as long as no nontrivial
strongly connected components occur. In our example automaton, the
states $8$, $11$ and $12$ are merged, and the states $10$ and $7$ also
are merged. This gives the automaton on the right of
Figure~\ref{BBC:fig:AZ1}.

Then a cycle which has all its  descendants minimized is checked for possible
minimization. This is done as follows: the \emph{weak signature} of a
state $p$ of a cycle is the signature obtained by replacing the name of
its successor in the cycle by a dummy symbol, say $\Box$. In our
example, the weak signatures of the states $3,5,9,6$ are respectively:
\begin{displaymath}
  -a8b\Box\,,\quad+a\Box b7\,,\quad-a8b\Box\,,\quad+a\Box b7\,.
\end{displaymath}
Here we write `${+}$' when the state is final, and `$-$' otherwise.

It is easily seen that the cycle is minimal if and only if the word
composed of the sequence of signatures is primitive. In our example,
the word is not primitive since it is a square, and the cycle can be
reduced by identifying states that are at corresponding positions in
the word, that is states $5$ and $6$ can be merged, and states $3$ and
$9$. This gives the automaton on the left of
Figure~\ref{BBC:fig:AZ2}.



Similarly, in order to check whether two (primitive) cycles can be
merged, one checks whether the words of their weak signatures are
conjugate. In our example, the cycles $2,4$ and $3,5$ have the
signatures
\begin{displaymath}
  +a\Box b7\,,-a8b\Box\quad\text{and}\quad -a8b\Box\,,+a\Box b7\,.
\end{displaymath}
These words are conjugate and the corresponding states can be
merged. This gives the automaton on the right of Figure~\ref{BBC:fig:AZ2}. This
automaton is minimal.

A basic argument for preserving the linearity of the algorithm is the
fact that the minimal conjugate of a word can be computed in linear
time. This can be done for instance by Booth's algorithm (see
\cite{Crochemore&Hancart&Lecroq:2007}). Thus, testing whether a cycle
is minimized takes time proportional to its length, and for each
minimized cycle, a canonical representative, namely the unique minimal
conjugate, which is a Lyndon word\index{Lyndon word}, can be computed
in time proportional to its length. The equality of two cycles then
reduces to the equality  of the two associated words. Finding
isomorphic cycles is accomplished by a lexicographic ordering of the
associated words, followed by a simple scan for equal words.

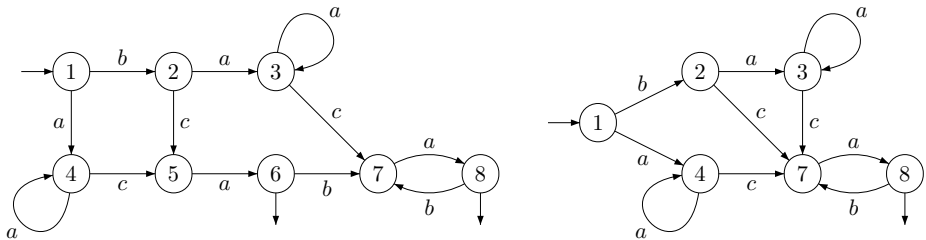
\begin{figure}
\centering\scalebox{0.85}{%
  \begin{picture}(78,36)(-10,-10)
    \gasset{Nh=6,Nadjust=w, Nadjustdist=2}
    \node[Nmarks=i,iangle=180](1)(0,16){$1$}
    \node(2)(16,16){$2$}
    \node(3)(32,16){$3$}
    \node(4)(0,0){$4$}
    \node(5)(16,0){$5$}
    \node[Nmarks=f,fangle=-90](6)(32,0){$6$}
    \node(7)(48,0){$7$}
    \node[Nmarks=f,fangle=-90](8)(64,0){$8$}
    \drawedge[ELside=r](1,4){$a$}\drawedge(1,2){$b$}
    \drawedge(2,3){$a$}\drawedge(2,5){$c$}
    \drawloop[loopangle=45](3){$a$}\drawedge(3,7){$c$}
    \drawloop[loopangle=-135](4){$a$}\drawedge[ELside=r](4,5){$c$}
    \drawedge[ELside=r](5,6){$a$}
    \drawedge[ELside=r](6,7){$b$}
    \drawedge[curvedepth=3](7,8){$a$}
    \drawedge[curvedepth=3](8,7){$b$}
  \end{picture}
\quad
  \begin{picture}(64,36)(-10,-10)
    \gasset{Nh=6,Nadjust=w, Nadjustdist=2}
    \node[Nmarks=i,iangle=180](1)(0,8){$1$}
    \node(2)(16,16){$2$}
    \node(3)(32,16){$3$}
    \node(4)(16,0){$4$}
    \node(7)(32,0){$7$}
    \node[Nmarks=f,fangle=-90](8)(48,0){$8$}
    \drawedge[ELside=r](1,4){$a$}\drawedge(1,2){$b$}
    \drawedge(2,3){$a$}\drawedge(2,7){$c$}
    \drawloop[loopangle=45](3){$a$}\drawedge(3,7){$c$}
    \drawloop[loopangle=-135](4){$a$}\drawedge[ELside=r](4,7){$c$}
    \drawedge[curvedepth=3](7,8){$a$}
    \drawedge[curvedepth=3](8,7){$b$}
  \end{picture}}
 \caption{The automaton on the left has one minimal cycle of
   height~$1$. By wrapping, states~$6$ and~$8$, and states~$5$ and~$7$
 are merged, respectively, giving the automaton on the right.}
\label{BBC:AZ3}
\end{figure}

A few words on wrapping: it may happen that states of distinct heights
in a simple automaton are equivalent. An example is given in
Figure~\ref{BBC:AZ3}. Indeed, states~$6$ and~$8$ have the same
signature and therefore are mergeable but have height~$1$ and~$0$,
respectively. This situation is typical: when states~$s$ and~$t$ are
mergeable and have distinct heights, and~$t$ belongs to a minimized
component of current height, then $s$ is a singleton component on a
path to the cycle of~$t$. Wrapping consists in detecting these states,
and in ``winding'' them around the cycle. In our example, both~$6$
and~$5$ are wrapped in the component of~$7$ and~$8$. In the algorithm
given above, a wrapping step is performed at each iteration, after the
minimization of the states and the cycles and before computing the
states and cycles of the next height. In our example, after the first
iteration, states~$3$ and~$4$ are mergeable. A second wrapping step
merges~$3$ and~$4$. These operations are reported in
Figure~\ref{BBC:AZ4}. A careful implementation can realize all these
operations in global linear time.

\begin{figure}
\centering\scalebox{0.85}{%
  \begin{picture}(64,30)(-10,-10)
    \gasset{Nh=6,Nadjust=w, Nadjustdist=2}
    \node[Nmarks=i,iangle=180](1)(0,8){$1$}
    \node(2)(16,16){$2$}
    \node(4)(16,0){$4$}
    \node(7)(32,8){$7$}
    \node[Nmarks=f,fangle=-90](8)(48,8){$8$}
    \drawedge[ELside=r](1,4){$a$}\drawedge(1,2){$b$}
    \drawedge(2,4){$a$}\drawedge(2,7){$c$}
    \drawloop[loopangle=-135](4){$a$}\drawedge[ELside=r](4,7){$c$}
    \drawedge[curvedepth=3](7,8){$a$}
    \drawedge[curvedepth=3](8,7){$b$}
  \end{picture}
\qquad\qquad
  \begin{picture}(64,30)(-10,-10)
    \gasset{Nh=6,Nadjust=w, Nadjustdist=2}
    \node[Nmarks=i,iangle=180](1)(0,8){$1$}
    \node(4)(16,8){$4$}
    \node(7)(32,8){$7$}
    \node[Nmarks=f,fangle=-90](8)(48,8){$8$}
    \drawedge(1,4){$a,b$}
    \drawloop[loopangle=-90](4){$a$}\drawedge(4,7){$c$}
    \drawedge[curvedepth=3](7,8){$a$}
    \drawedge[curvedepth=3](8,7){$b$}
  \end{picture}

}
 \caption{The automaton on the left has one minimal cycle of
   height~$1$. By wrapping, states~$6$ and~$8$, and states~$5$ and~$7$
 are merged, respectively, giving the automaton on the right.}
\label{BBC:AZ4}
\end{figure}
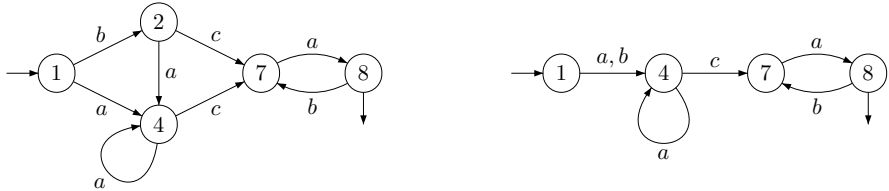





\subsection{Incremental minimization: the algorithm of Daciuk \textit{et al.}}

The algorithm presented in \cite{Daciuk&Mihov&Watson&Watson:2000} is
an incremental algorithm for the construction of a minimal automaton
for a given set of words that is lexicographically sorted. 

The algorithm is easy to implement and it is efficient: the
construction of an automaton recognizing a typical dictionary is done
in a few seconds. Table~\ref{BBC:table:1} was kindly communicated  by
S\'ebastien Paumier. It contains the space saving and the computation
time for dictionaries of various languages.

The algorithm described here is simple because the words are sorted.
There exist other incremental algorithms for the case of unsorted
sets. One of them will be described in the next section. Another
algorithm, called semi-incremental because it requires a final
minimization step, is given in \cite{Watson:2003}.

We start with some notation.
Let $\A=(Q,i,F)$ be a finite, acyclic, deterministic and
trim automaton.
We say that a word $x$ \emph{is in} the automaton $\A$ if
$i\cdot x$ is defined. In other words, $x$ is in $\A$ if $x$ is a
prefix of some word recognized by $\A$.
Let $w$ be a word to be added to the set recognized by an automaton
$\A$. The factorization 
\begin{displaymath}
  w=x\cdot y\,,
\end{displaymath}
where $x$ is the longest prefix of $w$ which is in $\A$, is called the
\emph{prefix-suffix decomposition}\index{prefix-suffix
  decomposition}\index{decomposition!prefix-suffix} of $w$. The word
$x$ (resp. $y$) is the \emph{common prefix}\index{common prefix}
(resp. \emph{corresponding suffix}) of $w$.

One has $x=\e$ if either $w=\e$ or $i\cdot a$ is undefined, where $a$ is
the initial letter of $w$. Similarly, $y=\e$ if $w$ itself is in
$\A$. If $y\ne\e$ and starts with the letter $b$, then $i\cdot
xb=\bot$.

The \textit{insertion} of a word $y$ at state $p$ is an operation that
is performed provided $y=\e$ or $p\cdot b=\bot$, where $b$ is the
initial letter of $y$. If $y=\e$, the insertion simply consists in adding
state $p$ to the set $F$ of final states.  If $y\ne\e$, set
$y=b_1b_2\cdots b_m$. The insertion consists in adding new states
$p_1,\ldots,p_m$ to $Q$, with the next state function defined by
$p\cdot b_1=p_1$ and $p_{i-1}\cdot b_i=p_i$ for $i=2,\ldots, m$.
Furthermore, $p_m$ is added to the set $F$ of final states.

Assume that the language recognized by $\A$ is not empty, and that the
word $w$ is lexicographically greater than all words in $\A$. Then $w$
is not in $\A$. So the common prefix $x$ of $w$ is strictly shorter than
$w$ and the corresponding suffix $y$ is nonempty.

\begin{table}
\centering
\begin{tabular}{|l|r|r||r|r|r|r|r|}%
\hline %
\multicolumn{1}{|c|}{\textbf{File}} & 
\multicolumn{1}{|c|}{\textbf{Lines}} & 
\multicolumn{1}{|c||}{\textbf{Text file}} & 
\multicolumn{3}{c|}{\textbf{Automaton}}& 
\multicolumn{2}{c|}{\textbf{Time}} \\%
\cline{4-8}  &  &         & \textbf{States} & \textbf{Trans.}
&\multicolumn{1}{|c|}{\textbf{Size}}
&\textbf{Revuz}&\textbf{Daciuk}
\\%
\hline%
\hline%
delaf-de   & 189878 & 12.5Mb &  57165 & 103362  & 1.45Mb & 4.22s   & 4.44s \\\hline%
delaf-en   & 296637 & 13.2Mb & 109965 & 224268  & 2.86Mb & 5.94s   & 6.77s \\\hline%
delaf-es   & 638785 & 35.4Mb &  56717 & 117417  & 1.82Mb & 10.61s   & 11.28s \\\hline%
delaf-fi   & 256787 & 24.6Mb & 124843 & 133288  & 4.14Mb & 6.40s   & 7.02s  \\\hline%
delaf-fr   & 687645 & 38.7Mb & 109466 & 240409  & 3.32Mb & 13.03s   & 14.14s  \\\hline%
delaf-gr   &1288218 & 83.8Mb & 228405 & 442977  & 7.83Mb & 28.33s   & 31.02s \\\hline%
delaf-it   & 611987 & 35.9Mb &  64581 & 161718  & 1.95Mb & 10.43s   & 11.46s \\\hline%
delaf-no   & 366367 & 23.3Mb &  75104 & 166387  & 2.15Mb & 6.86s   & 7.44s  \\\hline%
delaf-pl   &  59468 &  3.8Mb &  14128 &  20726  &  502Kb & 1.19s   &  1.30s \\\hline%
delaf-pt   & 454241 & 24.8Mb &  47440 & 115694  &  1.4Mb & 7.87s   & 8.45s \\\hline%
delaf-ru   & 152565 & 10.8Mb &  23867 &  35966  &  926Kb & 2.95s   &  3.17s \\\hline%
delaf-th   &  33551 &  851Kb &  36123 &  61357  &  925Kb & 0.93s   &  1.14s \\\hline%
\end{tabular}
\caption{Running time and space requirement for the computation of
  minimal automata (\textit{communication of S\'ebastien
    Paumier}).\label{BBC:table:size-performances}}  
\label{BBC:table:1}
\end{table}

The incremental algorithm works at follows. At each step, a new word
$w$ that is lexicographically greater than all previous ones is
inserted in the current automaton $\A$. First, the prefix-suffix
decomposition $w=xy$ of $w$, and the state $q=i\cdot x$ are
computed. Then the segment starting at $q$ of the path carrying the suffix
$y'$ of the previously inserted
word $w'$ is minimized by merging states with the
same signature. Finally, the suffix $y$ is inserted at state $q$.
The algorithm is given in Figure~\ref{BBC:alg:daciuk}.

The second step deserves a more detailed description. We observe first
that the word $x$ of the prefix-suffix decomposition $w=xy$ of $w$ is
in fact the greatest common prefix of $w'$ and $w$. Indeed, the word
$x$ is a prefix of some word recognized by $\A$ (here $\A$ is the
automaton before adding $w$), and since $w'$ is the greatest word in
$\A$, the word $x$ is a prefix of $w'$. Thus $x$ is a common prefix of
$w'$ and $w$. Next, if $x'$ is a common prefix of $w'$ and $w$, then
$x'$ is in $\A$ because it is a prefix of $w'$, and consequently $x'$
is a prefix of $x$ because $x$ is the longest prefix of $w$ in $\A$.
This shows the claim.

There are two cases for the merge.
If $w'$ is a prefix of $w$, then $w'=x$. In this case, there is no
minimization to be performed. 

If $w'$ is not a  prefix of $w$, then the paths for $w'$ and for $w$
share a common initial segment carrying the prefix $x$, from the
initial state to state $q=i\cdot x$. The minimization concerns the
states on the path $q\tto {y'} t'$ carrying the suffix $y'$ of the
factorization $w=xy'$ of $w'$. Each of the states in this path,
except the state $q$, will never be visited again in any insertion that
may follow, so they can be merged with previous states.

\begin{figure}
  \setlength{\commentspace}{6cm}
\hrule\smallskip\par
  \begin{algorithmic}
    \FUNC{DaciukEtAl$(\A)$}
    \FORALL{\algcomment{2.82}{Words are given in lexicographic order}$w$}
    \STATE\algcomment{1}{$x$ is the longest prefix of $w$ in $\A$}%
    $(x,y)\gets\textsc{PrefSuffDecomp}(w)$
    \STATE\algcomment{1}{$q$ is the state reached by reading $x$}$q\gets i\cdot x$
    \STATE\algcomment{1}{Minimize the states on this
      path}$\textsc{MinimizeLastPath}(q)$ 
    \STATE\algcomment{1}{Adds a path starting in $q$ and carrying $y$}%
    $\textsc{AddPath}(q,y)$ 
    \ENDFOR
  \end{algorithmic}
\hrule\smallskip
\caption{The incremental algorithm of Daciuk \textit{et al}.}\label{BBC:alg:daciuk}
\end{figure}

\begin{example}\label{BBC:example}
  We consider the sequence of words $(aa,aba,ba,bba,bc)$. The
  first two words give the automaton of
  Figure~\ref{BBC:fig:daciuk1}(a). Adding the word $ba$ permits the
  merge of states $2$ and $4$. The resulting automaton is given in 
Figure~\ref{BBC:fig:daciuk1}(b). After inserting $bba$, there is a
merge of states $6$ and $2$,  
see  Figure~\ref{BBC:fig:daciuk2}(a).
\end{example}

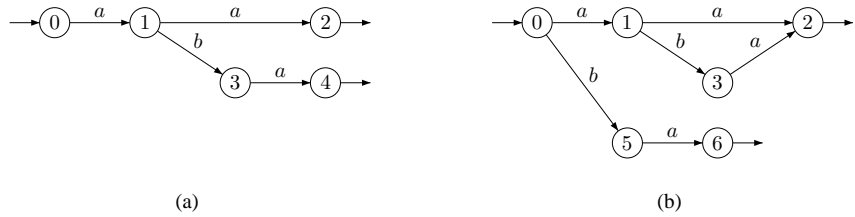
\begin{figure} 
  \centering
  \scalebox{0.8}{%
  \gasset{Nh=5,Nw=5}
  \begin{picture}(45,35)(0,-30)
    \node[Nmarks=i](0)(0,0){$0$}
    \node(1)(15,0){$1$}
    \node[Nmarks=f](2)(45,0){$2$}
    \node(3)(30,-10){$3$}
    \node[Nmarks=f](4)(45,-10){$4$}
    \drawedge(0,1){$a$}\drawedge(1,2){$a$}
    \drawedge(1,3){$b$}\drawedge(3,4){$a$}
    \node[Nframe=n](numero)(22,-30){(a)}
  \end{picture}\qquad\qquad\qquad\qquad\qquad
  \begin{picture}(45,35)(0,-30)
    \node[Nmarks=i](0)(0,0){$0$}
    \node(1)(15,0){$1$}
    \node[Nmarks=f](2)(45,0){$2$}
    \node(3)(30,-10){$3$}
    \drawedge(0,1){$a$}\drawedge(1,2){$a$}
    \drawedge(1,3){$b$}\drawedge(3,2){$a$}
    \node(5)(15,-20){$5$}
    \node[Nmarks=f](6)(30,-20){$6$}
    \drawedge(0,5){$b$}\drawedge(5,6){$a$}
    \node[Nframe=n](numero)(22,-30){(b)}
  \end{picture}}
  \caption{(a). The automaton for $aa,aba$. (b). The automaton for
    $aa,aba,ba$. Here  state $4$ has been merged with state $2$.} 
  \label{BBC:fig:daciuk1}
\end{figure}

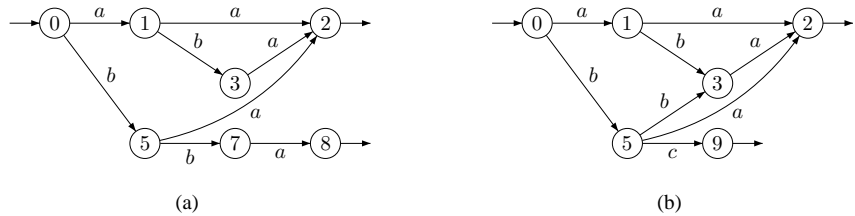
\begin{figure}
  \centering
  \scalebox{0.8}{
    \gasset{Nh=5,Nw=5}
  \begin{picture}(45,35)(0,-30)
    \node[Nmarks=i](0)(0,0){$0$}
    \node(1)(15,0){$1$}
    \node[Nmarks=f](2)(45,0){$2$}
    \node(3)(30,-10){$3$}
    \drawedge(0,1){$a$}\drawedge(1,2){$a$}
    \drawedge(1,3){$b$}\drawedge(3,2){$a$}
    \node(5)(15,-20){$5$}
    \drawedge(0,5){$b$}\drawedge[ELside=r,curvedepth=-4](5,2){$a$}
    \node(7)(30,-20){$7$}
    \node[Nmarks=f](8)(45,-20){$8$}
    \drawedge[ELside=r](5,7){$b$} 
    \drawedge[ELside=r](7,8){$a$}
    \node[Nframe=n](numero)(22,-30){(a)}
  \end{picture}\qquad\qquad\qquad\qquad\qquad
   \begin{picture}(45,35)(0,-30)
    \node[Nmarks=i](0)(0,0){$0$}
    \node(1)(15,0){$1$}
    \node[Nmarks=f](2)(45,0){$2$}
    \node(3)(30,-10){$3$}
    \drawedge(0,1){$a$}\drawedge(1,2){$a$}
    \drawedge(1,3){$b$}\drawedge(3,2){$a$}
    \node(5)(15,-20){$5$}
    \drawedge(0,5){$b$}
    \drawedge[ELside=r,curvedepth=-4](5,2){$a$}
    \drawedge(5,3){$b$} 
    \node[Nmarks=f](9)(30,-20){$9$}
    \drawedge[ELside=r](5,9){$c$} 
    \node[Nframe=n](numero)(22,-30){(b)}
  \end{picture}}
 \caption{(a). The automaton for $aa,aba,ba,aba,bba$. (b). The
   automaton for $aa,aba,ba,aba,bba,bc$. After inserting $bc$, 
   states $8$ and $2$ are merged, and then  states $7$ and $2$.} 
  \label{BBC:fig:daciuk2}
\end{figure}



%
%
\section{Dynamic minimization}\label{BBC:sec:dynamic}

\textit{Dynamic minimization} is the process of maintaining an
automaton minimal when insertions or deletions are performed.  

A solution for adding and for removing a word was proposed by Carrasco
and Forcada \cite{Carrasco&Forcada:2002}. It consists in an adaptation
of the usual textbook constructions for intersection and complement to
the special case where one of the languages is a single word. It
appears that the finiteness of the language $L$ plays no special role,
so we assume here that it is regular, not necessarily finite.  The
construction for adding a word has also been proposed
in~\cite{Sgarbas&Fakotakis&Kokkinakis:2003}, and
in~\cite{Daciuk&Mihov&Watson&Watson:2000} for acyclic automata. An
extension to general automata, and several other issues, are discussed
in~\cite{Daciuk:2004}.

We consider here, for lack of space, only deletion of a word from the
set recognized by an automaton, and minimization of the new automaton.

Let $\A=(Q, i, T)$ be the minimal automaton recognizing a language $L$
over the alphabet $A$, and let $w$ be a word in $L$. Denote by $\A_w$
the minimal automaton recognizing the complement $A^*\setminus w$. The
automaton has $n+2$ states, with $n=|w|$. Among them, there are $n+1$
states that are identified with the set $P$ of prefixes of $w$. The
last state is a sink state denoted $\bot$. An example is given in
Figure~\ref{BBC:fig:abab}.

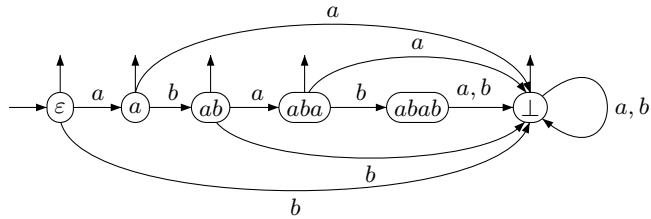
\begin{figure}
  \centering\small\unitlength=0.5mm
  \begin{picture}(125,50)(0,-25)
    \gasset{Nh=8,Nadjust=w,Nadjustdist=2,fangle=90}
    \node[Nmarks=if](e)(0,0){$\e$}
    \node[Nmarks=f](a)(20,0){$a$}
    \node[Nmarks=f](ab)(40,0){$ab$}
    \node[Nmarks=f](aba)(65,0){$aba$}
    \node(abab)(95,0){$abab$}
    \node[Nmarks=f](bot)(125,0){$\bot$}
    \drawedge(e,a){$a$}
    \drawedge(a,ab){$b$}
    \drawedge(ab,aba){$a$}
    \drawedge(aba,abab){$b$}
    \drawedge(abab,bot){$a,b$}
    \drawloop[loopangle=0](bot){$a,b$}
    \drawbpedge[ELside=r](e,-100,30,bot,-80,30){$b$}
    \drawbpedge(a,100,30,bot,80,30){$a$}
    \drawbpedge[ELside=r](ab,-90,18,bot,-90,18){$b$}
    \drawbpedge(aba,90,18,bot,90,18){$a$}
      \end{picture}
  \caption{The minimal automaton recognizing the complement of the
    word $abab$. Only the state $abab$ is not final.}
  \label{BBC:fig:abab}
\end{figure}
The language $L\setminus w$ is equal to $L\cap(A^*\setminus w)$, so it
is recognized by the trimmed part $B$ of the product automaton
$\A\times\A_w$. Its initial state is $(i,\e)$, and its states are of
three kinds.
\begin{itemize}
\item \emph{intact states}: these are states of the form $(q,\bot)$
  with $q\in Q$. They are called so because the language recognized at
  $(q,\bot)$ in $\B$ is the same as the language recognized at $q$ in
  $\A$: $L_\B(q,\bot)=L_\A(q)$.
\item \emph{cloned states}: these are accessible states $(q,x)$ with
  $x\in P$, so $x$ is a prefix of $w$. Since we require these states
  to be accessible, one has $q=i\cdot x$ in $\A$, and there is one such
  state for each prefix. The next-state function on these states is
  defined by
  \begin{displaymath}
    (q,x)\cdot a=
    \begin{cases}
      (q\cdot a,xa)&\text{if $xa\in P$,}\\(q\cdot a,\bot)&\text{otherwise.}
    \end{cases}
  \end{displaymath}
  Observe that $(i\cdot w,\bot)$ is an intact state because $w$ is
  assumed to be recognized by~$\A$.
\item \emph{useless states}: these are all states that are removed
  when trimming the product automaton.
\end{itemize}
Trimming consists here in removing the state $(i,\bot)$ if it is no
longer accessible, and the states reachable only from this state. For
this, one follows the path defined by $w$ and starting in $(i,\bot)$
and removes the states until one reaches a confluent state (that has
at least two incoming edges). The automaton obtained is minimal. 

The whole construction finally consists in keeping the initial
automaton, by renaming a state $q$ as $(q,\bot)$, adding a cloned
path, and removing state $(i,\bot)$ if it is no longer accessible, and
the states reachable only from this state.  

\begin{figure}
  \setlength{\commentspace}{6cm}
\hrule\smallskip\par
  \begin{algorithmic}
    \FUNC{RemoveIncremental$(w,\A)$}
    \STATE\algcomment{0}{Add a fresh path for $w$ in $\A$}%
    $\A'\gets\textsc{AddClonedPath}(w,\A)$
    \STATE\algcomment{0}{Return trimmed automaton}$\textsc{Trim}(\A')$ 
  \end{algorithmic}
\bigskip\par
  \begin{algorithmic}
    \FUNC{AddClonedPath$(a_1\cdots a_n,\A)$}
    \STATE\algcomment{0}{Add a fresh initial state $q_0$}%
    $p_0\gets\textsc{Initial}(\A);\ q_0\gets\textsc{Clone}(p_0)$
    \FOR{$i=1$ \TO $n$}
   \STATE\algcomment{1}{$q_i$ inherits the transitions of $p_i$}%
    $p_i\gets p_{i-1}\cdot a_i;\ q_i\gets\textsc{Clone}(p_i)$
    \STATE\algcomment{1}{This edge is redirected}$q_{i-1}\cdot
    a_i\gets q_i$ 
    \ENDFOR
    \STATE$\textsc{SetFinal}(q_n,\FALSE)$
  \end{algorithmic}
\hrule\smallskip
\caption{Removing the word $w$ from the language recognized by $\A$.}\label{BBC:alg:removeincremental}
\end{figure}

Of course, one may also use the textbook construction directly, that is
without taking advantage of the existence of the automaton given at
the beginning.  For this, one starts at the new initial state
$(i,\e)$ and one builds only the accessible part of the product
automaton. The method has complexity $O(n+|w|)$, where $n$ is the
number of states of the initial automaton, whereas the previous method
has only complexity $O(|w|)$.
\begin{figure}
  \centering\scalebox{0.85}{%
  \begin{picture}(60,24)(0,-8)
    \gasset{Nh=6,Nw=6,fangle=90}
    \node[Nmarks=i](0)(0,0){$0$}
    \node(1)(15,0){$1$}
    \node[Nmarks=f](2)(30,7){$2$}
    \node(3)(30,-7){$3$}
    \node(4)(45,7){$4$}
    \node[Nmarks=f](5)(45,-7){$5$}
    \node[Nmarks=f](6)(60,7){$6$}
    \drawedge(0,1){$a$}
    \drawedge(1,2){$b$}\drawedge(1,3){$c$}
    \drawedge(2,4){$a$}\drawedge(2,5){$c$}
    \drawedge(3,5){$b$}
    \drawedge[curvedepth=5](4,6){$b$}
    \drawedge[curvedepth=5](6,4){$a$}
 \end{picture}}
  \caption{The minimal automaton recognizing the language
  $L=(ab)^+\cup\{abc,acb\}$.}
  \label{BBC:fig:lang}
\end{figure}
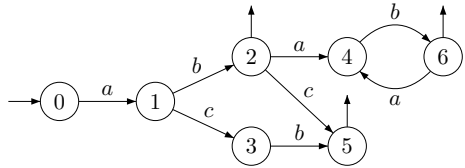
\begin{example}
  The automaton given in Figure~\ref{BBC:fig:lang} recognizes the language
  $L=(ab)^+\cup\{abc,acb\}$. The direct product with the automaton of
  Figure~\ref{BBC:fig:abab} is shown in Figure~\ref{BBC:fig:produit}. Observe
  that there are intact states that are not accessible from the new
  initial state $(0,\e)$. The minimal automaton is shown in
  Figure~\ref{BBC:fig:minimal}.  
\end{example}
\begin{figure}
  \centering\scalebox{0.85}{
  \begin{picture}(60,44)(0,-28)
    \gasset{Nh=6,Nw=8,fangle=90}
    \node[Nfill=y,fillgray=0.9](0)(0,0){$0,\bot$}
    \node[Nfill=y,fillgray=0.9](1)(15,0){$1,\bot$}
    \node[Nmarks=f,fillgray=0.9](2)(30,7){$2,\bot$}
    \node(3)(30,-7){$3,\bot$}
    \node(4)(45,7){$4,\bot$}
    \node[Nmarks=f](5)(45,-7){$5,\bot$}
    \node[Nmarks=f](6)(60,7){$6,\bot$}
    \drawedge(0,1){$a$}
    \drawedge(1,2){$b$}\drawedge(1,3){$c$}
    \drawedge(2,4){$a$}\drawedge(2,5){$c$}
    \drawedge(3,5){$b$}
    \drawedge[curvedepth=5](4,6){$b$}
    \drawedge[curvedepth=5](6,4){$a$}
    \def\vsize{-20}
      \node[Nmarks=i](e)(0,\vsize){$0,\e$}
      \node(a)(15,\vsize){$1,a$}
      \node[Nmarks=f,fangle=-90](ab)(30,\vsize){$2,ab$}
      \node[Nadjust=w](aba)(45,\vsize){$4,aba$}
      \node[Nadjust=w](abab)(60,\vsize){$6,abab$}
   \drawedge(e,a){$a$}
    \drawedge(a,ab){$b$}\drawedge(a,3){$c$}
    \drawedge(ab,aba){$a$}\drawedge(ab,5){$c$}
    \drawedge(aba,abab){$b$}
    \drawedge[ELside=r](abab,4){$a$}
   \end{picture}}
  \caption{The  automaton recognizing the  language
  $L=(ab)^+\cup\{abc,acb\}\setminus\{abab\}$. There are still
  unreachable states (shown in gray).}
  \label{BBC:fig:produit}
\end{figure}
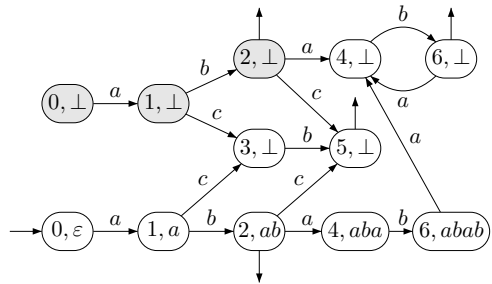
\begin{figure}[ht]
  \centering\scalebox{0.85}{
  \begin{picture}(75,25)(0,-8)
    \gasset{Nh=6,Nw=6,fangle=90}
    \node(3)(30,13){$3$}
    \node[Nmarks=f](4)(45,13){$5$}
    \node(7)(60,13){$7$}
    \node[Nmarks=f](8)(75,13){$8$}
    \node[Nmarks=i](e)(0,0){$0$}
    \node(a)(15,0){$1$}
    \node[Nmarks=f,fangle=-90](ab)(30,0){$2$}
    \node(aba)(45,0){$4$}
    \node(abab)(60,0){$6$}
    \drawedge(3,4){$b$}
    \drawedge[curvedepth=5](7,8){$b$}
    \drawedge[curvedepth=5](8,7){$a$}
   \drawedge(e,a){$a$}
    \drawedge(a,ab){$b$}\drawedge(a,3){$c$}
    \drawedge(ab,aba){$a$}\drawedge(ab,4){$c$}
    \drawedge(aba,abab){$b$}
    \drawedge[ELside=r](abab,7){$a$}
   \end{picture}}
  \caption{The  minimal automaton for the  language
  $L=(ab)^+\cup\{abc,acb\}\setminus\{abab\}$.}
  \label{BBC:fig:minimal}
\end{figure}
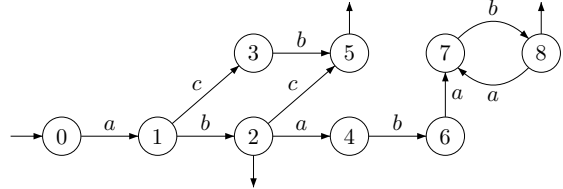


%
%

\section{Extensions and special cases}

In this section, we consider extensions of the minimization problem to
other classes of automata. The most important problem is to find a
minimal nondeterministic automaton recognizing a given regular
language. Other problems, not considered here, concern sets of
infinite words and the minimization of their accepting devices, and
the use of other kinds of automata known to be equivalent with respect
to their accepting capability, such as two-way automata or alternating
automata, see for instance~\cite{Perrin&Pin:2004}.

\subsection{Special automata}

We briefly mention here special cases where minimization plays a role.
It is well known that \emph{string matching} is closely related to the
construction of particular automata. If $w$ is a nonempty word over an
alphabet $A$, then searching for all occurrences of $w$ as a factor in a
text $t$ is equivalent to computing all prefixes of $t$ ending in $w$,
and hence to determining all prefixes of $t$ which are in the regular language
$A^*w$. The minimal automaton recognizing $A^*w$ has $n+1$ states,
where $n=|w|$, and can be constructed in linear time.  The automaton
has many interesting properties. For instance there are at most $2n$
edges, when one does not count edges ending in the initial state. This
is due to Imre Simon, see also \cite{Hancart:1993}. For a general
exposition, see e.g.  \cite{Crochemore&Hancart&Lecroq:2007}.
Extension of string matching to a finite set $X$ of patterns has been
done by Aho and Corasick. The associated automaton is called the
\emph{pattern matching machine}; it can be computed in time linear in
the sum of the lengths of the words in $X$. See again
\cite{Crochemore&Hancart&Lecroq:2007}. However, this automaton is not
minimal in general. Indeed, the number of states is the number of
distinct prefixes of words in $X$, and this may be greater than than
the number of states in the minimal automaton (consider for example
the set $X=\{ab,bb\}$ over the alphabet $A=\{a,b\}$). There are some
investigations on the complexity of minimizing Aho--Corasick automata,
see~\cite{AitMous&Bassino&Nicaud:2010}.

Another famous minimal automaton is the  \emph{suffix automaton}. This
is the minimal automaton recognizing all suffixes of a given
word. The number of states of the suffix automaton of a word of length
$n$ is less than $2n$, and the number of its edges is less than
$3n$. Algorithms for constructing suffix automata in linear time have
been given in \cite{Crochemore:1986} and \cite{Blumer&Blumer&Haussler:1985}, see again
\cite{Crochemore&Hancart&Lecroq:2007} for details.


%
%

\subsection{Nondeterministic automata}

A nondeterministic automaton is minimal%
\index{minimal!nondeterministic automaton}%
\index{automaton!minimal nondeterministic}%
\index{nondeterministic automaton!minimal} if it has the minimal
number of states among all automata recognizing the same language.
Nondeterministic automata are not unique. 
In
Figure~\ref{BBC:fig:dnondet}, we give two non-isomorphic
nondeterministic automata which are both smaller than the minimal
deterministic automaton recognizing the same language. This language
is $a(b^*\cup c^*)ab^+$. The example is derived from an automaton
given in \cite{Arnold&Dicky&Nivat:1992}.
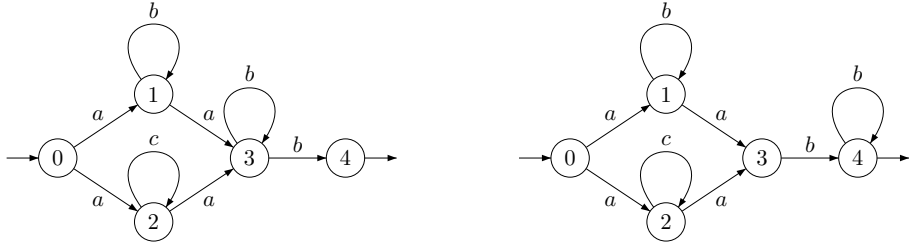
\begin{figure}[htb]  
  \centering
  \scalebox{0.85}{%
  \gasset{Nh=6,Nw=6}
  \begin{picture}(45,32)(0,-10)
    \node[Nmarks=i](0)(0,0){$0$}
    \node(1)(15,10){$1$}
    \node(2)(15,-10){$2$}
    \node(3)(30,0){$3$}
    \node[Nmarks=f](4)(45,0){$4$}
    \drawloop[loopangle=90](1){$b$}
    \drawloop[loopangle=90](2){$c$}
    \drawedge(0,1){$a$}\drawedge[ELside=r](0,2){$a$}
    \drawedge(1,3){$a$}\drawedge[ELside=r](2,3){$a$}
    \drawloop[loopangle=90](3){$b$}\drawedge(3,4){$b$}
  \end{picture}\qquad\qquad\qquad\qquad\qquad
  \begin{picture}(45,32)(0,-10)
    \node[Nmarks=i](0)(0,0){$0$}
    \node(1)(15,10){$1$}
    \node(2)(15,-10){$2$}
    \node(3)(30,0){$3$}
    \node[Nmarks=f](4)(45,0){$4$}
    \drawloop[loopangle=90](1){$b$}
    \drawloop[loopangle=90](2){$c$}
    \drawedge(0,1){$a$}\drawedge[ELside=r](0,2){$a$}
    \drawedge(1,3){$a$}\drawedge[ELside=r](2,3){$a$}
    \drawloop[loopangle=90](4){$b$}\drawedge(3,4){$b$}
  \end{picture}}
  \caption{Two non-isomorphic non-deterministic automata recognizing the
    set $a(b^*\cup c^*)ab^+$.} 
  \label{BBC:fig:dnondet}
\end{figure}

One might ask if there are simple conditions on the automata or on the
language that ensure that the minimal nondeterministic automaton is unique.
For instance, the  automata of
Figure~\ref{BBC:fig:deuxnondet} both recognize the same language, but
the second has a particular property that we will describe now.
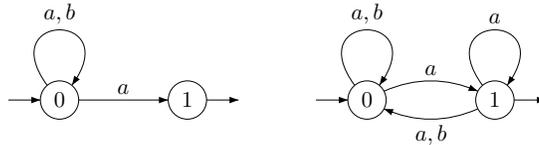
\begin{figure}[htb]
  \centering\scalebox{0.85}{
  \begin{picture}(20,15)(0,-3)
    \gasset{Nh=6,Nw=6}
    \node[Nmarks=i](0)(0,0){$0$}
    \node[Nmarks=f](1)(20,0){$1$}
    \drawloop[loopangle=90](0){$a,b$}
    \drawedge(0,1){$a$}
  \end{picture}\qquad\qquad\qquad\qquad
  \begin{picture}(20,15)(0,-3)
    \gasset{Nh=6,Nw=6}
    \node[Nmarks=i](0)(0,0){$0$}
    \node[Nmarks=f](1)(20,0){$1$}
    \drawloop[loopangle=90](0){$a,b$}
    \drawloop[loopangle=90](1){$a$}
    \drawedge[curvedepth=3](0,1){$a$}
    \drawedge[curvedepth=3](1,0){$a,b$}
  \end{picture}}
  \caption{Two nondeterministic automata recognizing the set of words
    ending with the letter $a$.}
  \label{BBC:fig:deuxnondet}
\end{figure}
The uniqueness of the minimal automaton in the deterministic case is
related to the fact that the futures of the states of such an
automaton are pairwise distinct, and that each future is some left
quotient of the language: for each state $q$, the language $L_q(\A)$
is equal to a set $y^{-1}L$, for some word $y$.

This characterization has been the starting point for investigating
similar properties of nondeterministic automata. Let us call a
(nondeterministic) automaton a \emph{residual
  automaton}\index{residual automaton} if the future of its states are
left quotients of the language; it has been shown in
\cite{Denis&Lemay&Terlutte:2001} that, among all residual automata
recognizing a given language, there is a unique residual automaton
having a minimal number of states; moreover, this automaton is
characterized by the fact that the set of its futures is the set of
the prime left quotients of the language, a left quotient being
\emph{prime} if it is not the union of other nonempty left quotients.
For instance, the automaton on the right of
Figure~\ref{BBC:fig:deuxnondet} has this property, since
$L_0=\{a,b\}^*a$ and $L_1=a^{-1}L_0=\e\cup\{a,b\}^*a$ and there are no
other nonempty left quotients. The automaton on the left of
Figure~\ref{BBC:fig:deuxnondet} is not residual since the future of
state~$1$ is not a left quotient.

The problem of converting a given nondeterministic automaton into a
minimal nondeterministic automaton is NP-hard, even over a unary
alphabet \cite{Jiang&Ravikumar:1993}. This applies also to unambiguous
automata \cite{Malcher:2004}. In  \cite{Bjorklund&Martens:2008}, these
results have been
extended  as follows. The authors define a class $\delta$NFA of
automata that are unambiguous, have at most two computations for each
string, and have at most one state with two outgoing transitions
carrying the same letter.  They show that minimization is NP-hard for
all classes of finite automata that include $\delta$NFA, and they show
that these hardness results can also be adapted to the setting of
unambiguous automata that can non-deterministically choose between two
start states, but are deterministic everywhere else.

Even approximating minimization of
nondeterministic automata is intractable, see
\cite{Gramlich&Schnitger:2005}. There is an algebraic framework that
allows one to represent and to compute all automata recognizing a given
regular language. The state of this theory, that goes back to
Kameda and Weiner ~\cite{Kameda&Weiner:1970},
has been described by Lombardy and Sakarovitch in a recent survey
paper~\cite{Lombardy&Sakarovitch:2007}. 

There is a well-known exponential blow-up from nondeterministic
automata to deterministic ones. The usual textbook example, already
given in the first section (the automaton on the left in
Figure~\ref{BBC:fig:nondetA}) shows that this blow-up holds also for
unambiguous automata, even if there is only one edge that causes the
nondeterminism.

It has been shown that any value of blow-up can be obtained, in the
following sense \cite{Jirasek&Jiraskova&Szabari:2007}: for all integers
$n, N$ with $n\le N\le 2^n$, there exists a minimal nondeterministic
automaton with $n$ states over a four-letter alphabet whose equivalent
minimal deterministic automaton has exactly $N$ states. This was
improved to ternary alphabets \cite{Jiraskova:2009}.

\nobreak
%
%

\section*{Acknowledgements}

We had several helpful discussions with Marie-Pierre B\'eal, Julien
David, Sylvain Lombardy, Wim Martens, Cyril Nicaud, S\'ebastien
Paumier, Jean-\'Eric Pin and Jacques Saka\-rovitch. We thank Narad
Rampersad for his careful reading of the text.

\bibliographystyle{abbrv}
\addcontentsline{toc}{section}{References}
\begin{footnotesize}
  \bibliography{abbrevs,MinimizationAutomata}

\newcommand{\noopsort}[1]{} \newcommand{\singleletter}[1]{#1}
  \newcommand{\etal}{et al.}
\begin{thebibliography}{10}

\bibitem{Aho&Hopcroft&Ullman:1974}
A.~Aho, J.~E. Hopcroft, and J.~D. Ullman.
\newblock {\em The design and analysis of computer algorithms}.
\newblock Addison-Wesley, 1974.

\bibitem{AitMous&Bassino&Nicaud:2010}
O.~AitMous, F.~Bassino, and C.~Nicaud.
\newblock Building the minimal automaton of {$A^*X$} in linear time, when {$X$}
  is of bounded cardinality.
\newblock In A.~Amir and L.~Parida, editors, {\em Combinatorial Pattern
  Matching, 21st Annual Symposium, CPM 2010}, volume 6129 of {\em Lecture Notes
  in Comput. Sci.}, pages 275--287. Springer-Verlag, 2010.

\bibitem{Almeida&Zeitoun:2008}
J.~Almeida and M.~Zeitoun.
\newblock Description and analysis of a bottom-up {DFA} minimization algorithm.
\newblock {\em Inform. Process. Lett.}, 107(2):52--59, 2008.

\bibitem{Arnold&Dicky&Nivat:1992}
A.~Arnold, A.~Dicky, and M.~Nivat.
\newblock A note about minimal non-deterministic automata.
\newblock {\em Bull. European Assoc. Theor. Comput. Sci.}, 47:166--169, 1992.

\bibitem{Baclet&Pagetti:2006}
M.~Baclet and C.~Pagetti.
\newblock Around {H}opcroft's algorithm.
\newblock In {\em 11th Conference on Implementation Application Automata
  (CIAA)}, volume 4094 of {\em Lecture Notes in Comput. Sci.}, pages 114--125.
  Springer-Verlag, 2006.

\bibitem{Bassino&David&Nicaud:2011}
F.~Bassino, J.~David, and C.~Nicaud.
\newblock Average case analysis of {M}oore's state minimization algorithm.
\newblock {\em submitted}.

\bibitem{Bassino&David&Nicaud:2009}
F.~Bassino, J.~David, and C.~Nicaud.
\newblock On the average complexity of {M}oore's state minimization algorithm.
\newblock In S.~Albers and J.-Y. Marion, editors, {\em STACS 2009, Proc. 26th
  Symp. Theoretical Aspects of Comp. Sci.}, volume 09001 of {\em Dagstuhl
  Seminar Proceedings}, pages 123--134. Schloss Dagstuhl - Leibniz-Zentrum fuer
  Informatik, 2009.

\bibitem{Beal&Crochemore:2007}
M.-P. B\'eal and M.~Crochemore.
\newblock Minimizing local automata.
\newblock In G.~Caire and M.~Fossorier, editors, {\em IEEE International
  Symposium on Information Theory (ISIT'07)}, number 07CH37924C, pages
  1376--1380. IEEE, 2007.

\bibitem{Beal&Crochemore:2008}
M.-P. B\'eal and M.~Crochemore.
\newblock Minimizing incomplete automata.
\newblock In {\em Workshop on Finite State Methods and Natural Language
  Processing 2008}. Ispra, september 2008.

\bibitem{Beauquier&Berstel&Chretienne:1992}
D.~Beauquier, J.~Berstel, and P.~Chr\'etienne.
\newblock {\em {\'E}l\'ements d'algorithmique}.
\newblock Masson, 1992.

\bibitem{Berstel&Boasson&Carton:2009b}
J.~Berstel, L.~Boasson, and O.~Carton.
\newblock Continuant polynomials and worst-case behavior of {H}opcroft's
  minimization algorithm.
\newblock {\em Theoret. Comput. Sci.}, 410(30-32):2811--2822, 2009.

\bibitem{Berstel&Boasson&Carton&Fagnot:2010}
J.~Berstel, L.~Boasson, O.~Carton, and I.~Fagnot.
\newblock {S}turmian trees.
\newblock {\em Theory Comput. Systems}, 46:443--478, 2010.

\bibitem{Berstel&Carton:2004}
J.~Berstel and O.~Carton.
\newblock On the complexity of {H}opcroft's state minimization algorithm.
\newblock In {\em 9th Conference on Implementation Application Automata
  (CIAA)}, volume 3317 of {\em Lecture Notes in Comput. Sci.}, pages 35--44.
  Springer-Verlag, 2004.

\bibitem{Berstel&Perrin:2005}
J.~Berstel and D.~Perrin.
\newblock Algorithms on words.
\newblock In M.~Lothaire, editor, {\em Algebraic combinatorics on words},
  volume~90 of {\em Encyclopedia of Mathematics and its Applications},
  chapter~1, pages 1--100. Cambridge University Press, 2002.

\bibitem{Bjorklund&Martens:2008}
H.~Bj{\"o}rklund and W.~Martens.
\newblock The tractability frontier for {NFA} minimization.
\newblock In L.~Aceto, I.~Damg{\aa}rd, L.~A. Goldberg, M.~M. Halld{\'o}rsson,
  A.~Ing{\'o}lfsd{\'o}ttir, and I.~Walukiewicz, editors, {\em Automata,
  Languages and Programming, 35th International Colloquium, ICALP 2008}, volume
  5126 of {\em Lecture Notes in Comput. Sci.}, pages 27--38. Springer-Verlag,
  2008.

\bibitem{Blum:1996}
N.~Blum.
\newblock A {$O(n\log n)$} implementation of the standard method for minimizing
  $n$-state finite automata.
\newblock {\em Inform. Process. Lett.}, 57:65--69, 1996.

\bibitem{Blumer&Blumer&Haussler:1985}
A.~Blumer, J.~A. Blumer, D.~Haussler, A.~Ehrenfeucht, M.~T. Chen, and J.~I.
  Seiferas.
\newblock The smallest automaton recognizing the subwords of a text.
\newblock {\em Theoret. Comput. Sci.}, 40:31--55, 1985.

\bibitem{Brzozowski:1963}
J.~A. Brzozowski.
\newblock Canonical regular expressions and minimal state graphs for definite
  events.
\newblock In {\em Proc. {S}ympos. {M}ath. {T}heory of {A}utomata ({N}ew {Y}ork,
  1962)}, pages 529--561. Polytechnic Press of Polytechnic Inst. of Brooklyn,
  Brooklyn, N.Y., 1963.

\bibitem{Carrasco&Forcada:2002}
R.~C. Carrasco and M.~L. Forcada.
\newblock Incremental construction and maintenance of minimal finite-state
  automata.
\newblock {\em Comput. Linguist.}, 28(2), 2002.

\bibitem{Castiglione&Restivo&Sciortino:2007}
G.~Castiglione, A.~Restivo, and M.~Sciortino.
\newblock Circular words and automata minimization.
\newblock In P.~Arnoux, N.~B\'edaride, and J.~Cassaigne, editors, {\em Words
  2007}, pages 79--89. Institut de Math\'ematiques de Luminy, 17--21 september
  2007.

\bibitem{Castiglione&Restivo&Sciortino:2008}
G.~Castiglione, A.~Restivo, and M.~Sciortino.
\newblock Hopcroft's algorithm and cyclic automata.
\newblock In C.~Mart\'{\i}n-Vide, F.~Otto, and H.~Fernau, editors, {\em
  Language and Automata Theory and Applications}, volume 5196 of {\em Lecture
  Notes in Comput. Sci.}, pages 172--183. Springer-Verlag, 2008.
\newblock Tarragona, Spain, March 13-19, 2008. Revised Papers.

\bibitem{Castiglione&Restivo&Sciortino:2009}
G.~Castiglione, A.~Restivo, and M.~Sciortino.
\newblock Circular {S}turmian words and {H}opcroft's algorithm.
\newblock {\em Theoret. Comput. Sci.}, 410:4372--4381, 2009.

\bibitem{Crochemore:1986}
M.~Crochemore.
\newblock Transducers and repetitions.
\newblock {\em Theoret. Comput. Sci.}, 45:63--86, 1986.

\bibitem{Crochemore&Hancart&Lecroq:2007}
M.~Crochemore, C.~Hancart, and T.~Lecroq.
\newblock {\em Algorithms on strings}.
\newblock Cambridge University Press, 2007.

\bibitem{Daciuk:2002}
J.~Daciuk.
\newblock Comparison of construction algorithms for minimal, acyclic,
  deterministic finite-state automata from sets of strings.
\newblock In J.-M. Champarnaud and D.~Maurel, editors, {\em 7th Conference on
  Implementation Application Automata (CIAA)}, volume 2608 of {\em Lecture
  Notes in Comput. Sci.}, pages 255--261. Springer-Verlag, 2002.

\bibitem{Daciuk:2004}
J.~Daciuk.
\newblock Comments on "{I}ncremental construction and maintenance of minimal
  finite-state automata" by {Rafael C. Carrasco} and {Mikel L. Forcada}.
\newblock {\em Comput. Linguist.}, 30(2):227--235, 2004.

\bibitem{Daciuk&Mihov&Watson&Watson:2000}
J.~Daciuk, S.~Mihov, B.~W. Watson, and R.~E. Watson.
\newblock Incremental construction of minimal acyclic finite-state automata.
\newblock {\em Comput. Linguist.}, 26(1):3--16, april 2000.

\bibitem{David:2010}
J.~David.
\newblock The average complexity of {M}oore's state minimization algorithm is
  $o(n\log\log n)$.
\newblock In P.~Hlinen{\'y} and A.~Kucera, editors, {\em Mathematical
  Foundations of Computer Science 2010, 35th International Symposium, MFCS
  2010}, volume 6281 of {\em Lecture Notes in Comput. Sci.}, pages 318--329.
  Springer-Verlag, 2010.

\bibitem{David:2010b}
J.~David.
\newblock {\em G\'en\'eration al\'eatoire d'automates et analyse d'algorithmes
  de minimisation}.
\newblock Doctorat d'unversit\'e, University Paris-Est, september 2010.

\bibitem{Denis&Lemay&Terlutte:2001}
F.~Denis, A.~Lemay, and A.~Terlutte.
\newblock Residual finite state automata.
\newblock In {\em STACS 2001, Proc. 18th Symp. Theoretical Aspects of Comp.
  Sci.}, volume 2010 of {\em Lecture Notes in Comput. Sci.}, pages 144--157.
  Springer-Verlag, 2001.

\bibitem{Gramlich&Schnitger:2005}
G.~Gramlich and G.~Schnitger.
\newblock Minimizing {NFA}'s and regular expressions.
\newblock In {\em STACS 2005, Proc. 22th Symp. Theoretical Aspects of Comp.
  Sci.}, volume 3404 of {\em Lecture Notes in Comput. Sci.}, pages 399--411.
  Springer-Verlag, 2005.

\bibitem{Gries:1973}
D.~Gries.
\newblock Describing an algorithm by {H}opcroft.
\newblock {\em Acta Informatica}, 2:97--109, 1973.

\bibitem{Hancart:1993}
C.~Hancart.
\newblock On {S}imon's string searching algorithm.
\newblock {\em Inform. Process. Lett.}, 47(2):95--99, 1993.

\bibitem{Hopcroft:1971}
J.~E. Hopcroft.
\newblock An $n \log n$ algorithm for minimizing states in a finite automaton.
\newblock In Z.~Kohavi and A.~Paz, editors, {\em Theory of {M}achines and
  {C}omputations}, pages 189--196. Academic Press, 1971.

\bibitem{Hopcroft&Ullman:1979}
J.~E. Hopcroft and J.~D. Ullman.
\newblock {\em Introduction to automata theory, languages, and computation}.
\newblock Addison-Wesley Publishing Co., Reading, Mass., 1979.
\newblock Addison-Wesley Series in Computer Science.

\bibitem{Jiang&Ravikumar:1993}
T.~Jiang and B.~Ravikumar.
\newblock Minimal {NFA} problems are hard.
\newblock {\em SIAM J. Comput.}, 22(6):1117--1141, 1993.

\bibitem{Jirasek&Jiraskova&Szabari:2007}
J.~Jir{\'a}sek, G.~Jir{\'a}skov{\'a}, and A.~Szabari.
\newblock Deterministic blow-up of minimal nondeterministic finite automata
  over a fixed alphabet.
\newblock In J.~Karhum\"aki and A.~Lepist\"o, editors, {\em Developments in
  Language Theory}, volume 4588 of {\em Lecture Notes in Comput. Sci.}, pages
  254--265. Springer-Verlag, 2007.

\bibitem{Jiraskova:2009}
G.~Jir{\'a}skov{\'a}.
\newblock Magic numbers and ternary alphabet.
\newblock In V.~Diekert and D.~Nowotka, editors, {\em Developments in Language
  Theory}, volume 5583 of {\em Lecture Notes in Comput. Sci.}, pages 300--311.
  Springer-Verlag, 2009.

\bibitem{Kameda&Weiner:1970}
T.~Kameda and P.~Weiner.
\newblock On the state minimization of nondeterministic finite automata.
\newblock {\em IEEE Trans. Comput.}, C-19(7):617--627, 1970.

\bibitem{Knuutila:2001}
T.~Knuutila.
\newblock Re-describing an algorithm by {H}opcroft.
\newblock {\em Theoret. Comput. Sci.}, 250:333--363, 2001.

\bibitem{Krivol:1991}
S.~L. Krivol.
\newblock Algorithms for minimization of finite acyclic automata and pattern
  matching in terms.
\newblock {\em Cybernetics}, 27:324-- 331, 1991.
\newblock translated from Kibernetika, No 3, May-June 1991, pp. 11--16.

\bibitem{Lombardy&Sakarovitch:2007}
S.~Lombardy and J.~Sakarovitch.
\newblock The universal automaton.
\newblock In E.~Gr{\"a}del, J.~Flum, and T.~Wilke, editors, {\em Logic and
  automata: History and perspectives}, volume~2 of {\em Texts in Logic and
  Games}, pages 467--514. Amsterdam University Press, 2007.

\bibitem{Malcher:2004}
A.~Malcher.
\newblock Minimizing finite automata is computationally hard.
\newblock {\em Theoret. Comput. Sci.}, 327(3):375--390, 2004.

\bibitem{Moore:1956}
E.~F. Moore.
\newblock Gedanken experiments on sequential machines.
\newblock In C.~E. Shannon and J.~McCarthy, editors, {\em Automata Studies},
  pages 129--153. Princeton Universty Press, 1956.

\bibitem{Paun&Paun&Rodriguez&Paton:2009}
A.~Pa{\u u}n, M.~Pa{\u u}n, and A.~Rodr{\'{\i}}guez-Pat{\'o}n.
\newblock On {H}opcroft's minimization technique for {DFA} and {DFCA}.
\newblock {\em Theoret. Comput. Sci.}, 410:2424--2430, 2009.

\bibitem{Perrin&Pin:2004}
D.~Perrin and J.-E. Pin.
\newblock {\em Infinite words, automata, semigroups, logic and games}.
\newblock Elsevier, 2004.

\bibitem{Revuz:1992}
D.~Revuz.
\newblock Minimisation of acyclic deterministic automata in linear time.
\newblock {\em Theoret. Comput. Sci.}, 92:181--189, 1992.

\bibitem{Sakarovitch:2009}
J.~Sakarovitch.
\newblock {\em Elements of automata theory}.
\newblock Cambridge University Press, 2009.

\bibitem{Sgarbas&Fakotakis&Kokkinakis:2003}
K.~N. Sgarbas, N.~D. Fakotakis, and G.~K. Kokkinakis.
\newblock Optimal insertion in deterministic {DAWG}s.
\newblock {\em Theoret. Comput. Sci.}, 301:103--117, 2003.

\bibitem{Valmari&Lehtinen:2008}
A.~Valmari and P.~Lehtinen.
\newblock Efficient minimization of {DFA}s with partial transition.
\newblock In S.~Albers and P.~Weil, editors, {\em STACS 2008, Proc. 25th Symp.
  Theoretical Aspects of Comp. Sci.}, volume 08001 of {\em Dagstuhl Seminar
  Proceedings}, pages 645--656. Schloss Dagstuhl - Leibniz-Zentrum fuer
  Informatik, 2008.

\bibitem{Watson:2003}
B.~W. Watson.
\newblock A new algorithm for the construction of minimal acyclic {DFA}s.
\newblock {\em Sci. Comput. Program.}, 48(2-3):81--97, 2003.

\end{thebibliography}
\end{footnotesize}

\newpage
%
%
\begin{abstract} 
  This chapter is concerned with the design and analysis of algorithms
  for minimizing finite automata. Getting a minimal automaton is a
  fundamental issue in the use and implementation of finite automata
  tools in frameworks like text processing, image analysis, linguistic
  computer science, and many other applications.

  Although minimization algorithms are rather old, there were recently
  some new developments which are explained or sketched in this
  chapter.

  There are two main families of minimization algorithms. The first
  are by a sequence of refinements of a partition of the set of
  states, the second by a sequence of fusions or merges of states.
  Hopcroft's and Moore's algorithms belong to the first family. The
  linear-time minimization of acyclic automata of Revuz belongs to the
  second family.

  We describe, prove and analyze Moore's and Hopcroft's algorithms.
  One of our studies is the comparison of these algorithms. It appears
  that they quite different both in behavior and in complexity. In
  particular, we show that it is not possible to simulate the
  computations of one of the algorithm by the other.  We also state
  the results known about average complexity of both algorithms.

  We describe the minimization algorithm by fusion for acyclic
  automata, and for a more general class of automata, namely those
  recognizing regular languages having  polynomially bounded number of
  words. 

  Finally, we consider briefly incremental algorithms for building
  minimal automata for finite sets of words. We consider the updating
  of a minimal automaton when a word is added or removed.
\end{abstract}


\printindex
\end{document}